\documentclass[11pt,letterpaper]{article}
\usepackage{color,ifpdf,latexsym,graphicx,url,amsmath,enumerate}
\pdfoutput=1

\usepackage{CJKutf8} 

\title{Pachinko}
\author{%
  Hugo A. Akitaya
\and
  Erik D. Demaine
\and
  Martin L. Demaine
\and
  Adam Hesterberg
\and
  Ferran Hurtado
\and
  Jason S. Ku
\and
  Jayson Lynch
}
\date{}

\newif\ifabstract
\abstracttrue
\newif\iffull
\ifabstract \fullfalse \else \fulltrue \fi

\usepackage
  [breaklinks,bookmarks,bookmarksnumbered,bookmarksopen,bookmarksopenlevel=2]
  {hyperref}
{\makeatletter \hypersetup{pdftitle={\@title}}}

\advance \topmargin by -\headheight
\advance \topmargin by -\headsep
\evensidemargin \oddsidemargin
{\makeatletter
 \gdef\xxxmark{%
   \expandafter\ifx\csname @mpargs\endcsname\relax 
     \expandafter\ifx\csname @captype\endcsname\relax 
       \marginpar{xxx}
     \else
       xxx 
     \fi
   \else
     xxx 
   \fi}
 \gdef\xxx{\@ifnextchar[\xxx@lab\xxx@nolab}
 \long\gdef\xxx@lab[#1]#2{\textbf{[\xxxmark #2 ---{\sc #1}]}}
 \long\gdef\xxx@nolab#1{\textbf{[\xxxmark #1]}}
}

{\makeatletter \gdef\fps@figure{!htbp}}

\let\realbfseries=\bfseries
\def\bfseries{\realbfseries\boldmath}

\newtheorem{theorem}{Theorem}
\newtheorem{lemma}[theorem]{Lemma}

\newtheorem{open}{Open Problem}

\newenvironment{proof}{\noindent\textbf{Proof: }\ignorespaces}
  {\hspace*{\fill}$\Box$\medskip}

\let\epsilon=\varepsilon

\begin{document}
\maketitle

\begin{abstract}
  Inspired by the Japanese game Pachinko, we study simple 
  (perfectly ``inelastic'' collisions) 
  dynamics of a unit ball falling amidst point obstacles
  (\emph{pins}) in the plane.
  A classic example is that a checkerboard grid of pins produces the binomial
  distribution, but what probability distributions result
  from different pin placements?
  In the 50-50 model, where the pins form a subset of this grid,
  not all probability distributions are possible, but
  surprisingly the uniform distribution is possible for $\{1,2,4,8,16\}$
  possible drop locations.
  Furthermore, every probability distribution can be approximated arbitrarily
  closely, and every dyadic probability distribution can be divided by a
  suitable power of $2$ and then constructed exactly
  (along with extra ``junk'' outputs).
  In a more general model, if a ball hits a pin off center, it falls left or
  right accordingly.  Then we prove a universality result:
  any distribution of $n$ dyadic probabilities, each specified by $k$ bits,
  can be constructed using $O(n k^2)$ pins,
  which is close to the information-theoretic lower bound of $\Omega(n k)$.
\end{abstract}

\begin{center}
\sl In memory of our friend Ferran Hurtado. 
\end{center}


\section{Introduction}




\begin{figure}[b]
  \centering
  \includegraphics[width=0.66\linewidth]{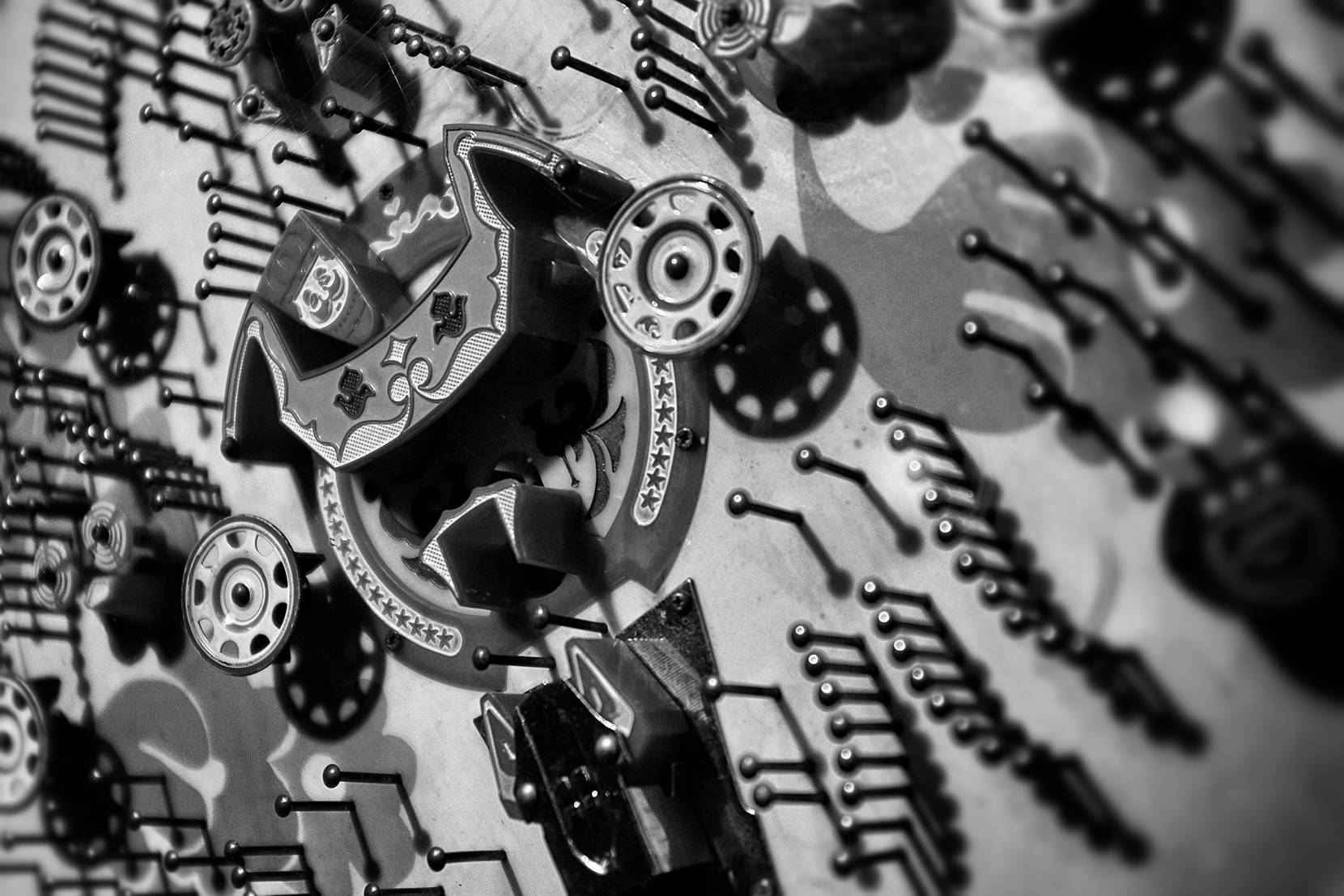}
  \caption{Close-up of Pachinko machine.  Photo by Neil Heeney, 2012,
    used with permission.
    \url{http://www.flickr.com/photos/heeney/8188927393}}
  \label{real pachinko 1}
  %
  %
\end{figure}

Pachinko \cite{Wiki,JapanZone} is a popular mechanical gambling game found in
tens of thousands of arcade parlors throughout Japan.
The player fires Pachinko balls (ball bearings) into a vertical,
nearly two-dimensional area filled with an
array of horizontal pins, spinners, winning pockets which reward the player
with more balls, etc.  See Figures \ref{real pachinko 1} and \ref{real pachinko 2}.
Since their invention in the 1920s until the 1980s,
balls launched using a mechanical flipper, similar to pinball machines,
while more recent Pachinko machines feature electrically controlled automatic
ball launching and slot-machine elements.  The name ``Pachinko''
(\begin{CJK}{UTF8}{goth}パチンコ\end{CJK})
comes from the Japanese word ``pachi pachi''
(\begin{CJK}{UTF8}{goth}パチパチ\end{CJK})
which imitates the sound of (in this case) metal
balls hitting metal pins.

\begin{figure}[t]
  \centering
  \begin{minipage}[t]{0.47\textwidth}
    \centering
    \includegraphics[trim=25 0 0 0,clip,width=\linewidth]{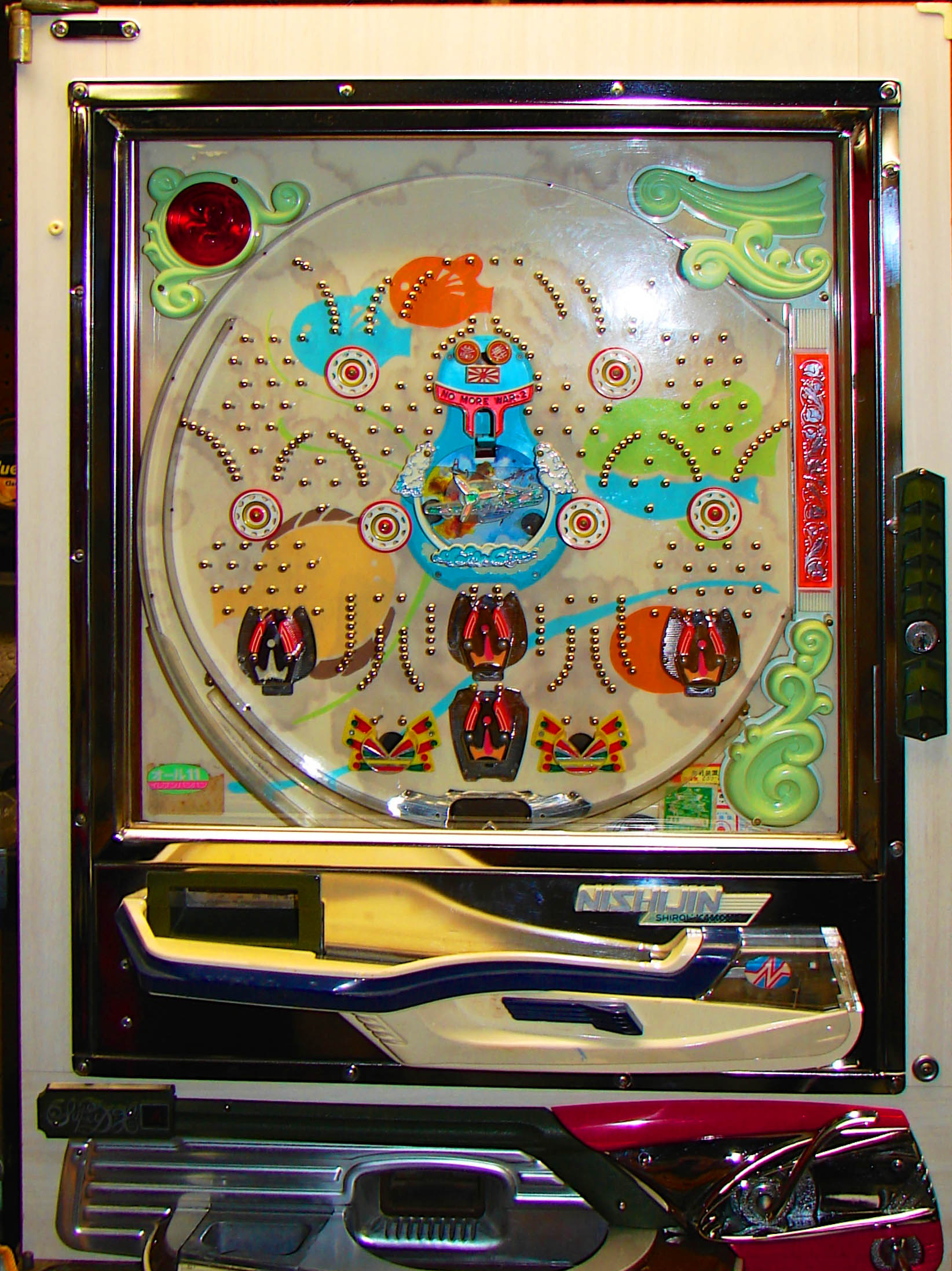}
    \caption[Nishijin mechanical Pachinko machine from the mid-20th century.
      Photo by Daniel Reed, 2008, used with permission.]
      {Nishijin mechanical Pachinko machine from the mid-20th century.
      Photo by Daniel Reed, 2008, used with permission.\footnotemark 
      }
    \label{real pachinko 2}
  \end{minipage}\hfill
  \begin{minipage}[t]{0.5\textwidth}
    \centering
    \includegraphics[trim=150 300 0 0,clip,width=\linewidth]{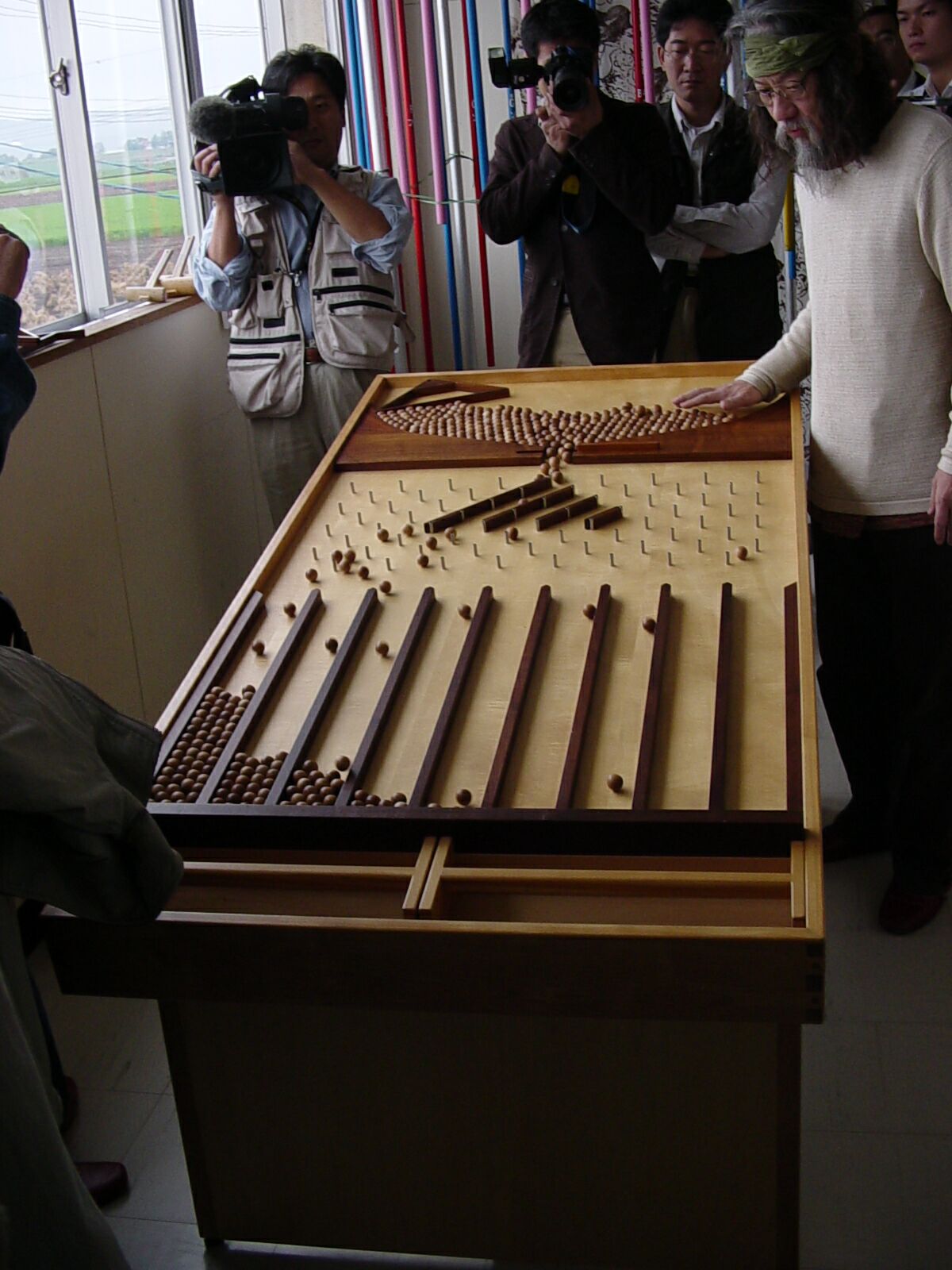}
    \caption{Jin Akiyama (right) and his mathematical Pachinko device.
      Photo taken by F. Hurtado at KyotoCGGT in Kyoto, 2007.}
    \label{akiyama}
  \end{minipage}
\end{figure}

In this paper,\footnotetext{\url{http://faculty.ccp.edu/faculty/dreed/Campingart/pachinko/}} we study an idealized geometry and dynamics of a simple form of
Pachinko: a single ball in a vertical plane falling through an arrangement of
pins, modeled as points.

Our study is motivated by a mathematical illustration/toy built by Jin Akiyama,
exhibited at KyotoCGGT 2007 and written about in 2008 \cite{Akiyama-Ruiz-2008};
see Figure~\ref{akiyama}.
This device features a regular, checkerboard grid of pins, with a reservoir of
wooden balls above, and a bank of containers below which allow a visual measurement
of how many balls go where.
If we model a ball hitting a pin as having a 50-50 chance of going left or
right, then the balls produce a binomial distribution,
which we visually recognize as a bell curve.
A similar illustration is in Eameses' famous Mathematica exhibit at the Boston
Museum of Science (since 1961) \cite{Eames}.

Akiyama's device can be augmented by wooden slats which force balls
to go to a particular direction (left or right) when hitting certain pins,
dramatically affecting the resulting distribution of balls
(as in Figure~\ref{akiyama}).
This type of forced ramp is common in real Pachinko machines as well,
simulated by an array of closely spaced pins
(see Figures \ref{real pachinko 1} and~\ref{real pachinko 2}).

The central problem addressed in this paper is what distributions of balls can
result purely from pins.  We consider three models/restrictions on the pins.
In the \emph{50-50 model}, a ball is only allowed to hit a pin dead on, and
thus fall to the left or right of the pin with equal probability.  In the
\emph{general model}, pins can be closely spaced, effectively allowing pins to
force the direction of the ball.  In both cases, we imagine the ball only
falling, not bouncing off of pins---effectively modeling perfectly ``inelastic" collisions---and
measure the probability of a ball reaching each column on the bottom
(\emph{outputs}).
Because the only random element in our Pachinko models is a 50-50 pin, all
output probabilities are \emph{dyadic}, i.e., of the form $i/2^k$ for
integers $i$ and~$k$.
Such probabilities can be represented finitely in binary, e.g., $0.10010111$. 

While this work is motivated by modeling Pachinko, the theory developed has 
real-world application in the design of liquid distribution systems. 
Given an input stream of liquid, the goal is to develop a network of 
pipes to output a specified amount of liquid into many different containers,
where each pipe junction equally distributes liquid between two outgoing children,
something easy to manufacture in practice.
The specific case of generating uniform distributions is a passive solution
to distributing equal amounts of liquid into bottles, 
and we provide uniform distribution networks for some
practical numbers of outputs.

\paragraph{Our results.}

In the general model, we prove a universality result: any distribution of $n$
dyadic probabilities (summing to 1) can be produced by a polynomial-size
arrangement of pins.  If all the probabilities can be specified in binary by
$k$ bits (i.e., the probabilities can be written with common denominator
$2^k$), then we show a constructive upper bound of $O(n k^2)$ pins and an
information-theoretic lower bound of $\Omega(n k)$.  These results leave a gap
of $\Theta(k)$ as an intriguing open problem.

The 50-50 model proves much more interesting because not all probability
distributions are possible.  For example, the only possible two-probability
distribution is $\left\langle \frac12,0, \frac12 \right\rangle$.  Nonetheless, we prove
several strong positive results:
\begin{enumerate}

\item For every probability distribution $\langle p_1, p_2, \dots, p_n \rangle$,
and for every $\epsilon > 0$, there is a 50-50 construction producing an output
probability within $\epsilon$ of each $p_i$, as well as extra outputs with
probability less than $\epsilon$.

\item For every dyadic probability distribution $\langle p_1, p_2, \dots, p_n\rangle$,
there is a constant scale factor $\alpha = 1/2^j$ and a 50-50 construction
producing output probability $\alpha p_i$ (i.e., a shifted version of $p_i$)
for each $i$, as well as additional garbage outputs.

\item While the uniform distribution is intuitively very difficult to produce
in the 50-50 model, we show that it is in fact possible for
$n \in \{1,2,4,8,16\}$ outputs, and conjecture that it is possible for all
$n = 2^k$.

\end{enumerate}

Because not all dyadic probability distributions are possible in the 50-50
model, the first two results are in some sense best possible.  Nonetheless, a
natural open question is to characterize exactly which probability
distributions are possible in the 50-50 model.  In particular, is every dyadic
probability constructible, ignoring all other output probabilities?  We
conjecture so, but currently only know (by the first result) that every
probability can be approximated arbitrarily closely.


\section{Models}

In this section, we define a simple formal model for Pachinko dynamics,
and then specialize to the 50-50 model and another simple grid model.
The environment is a planar arrangement of fixed point obstacles
called \emph{pins}; and a disk called a \emph{ball} which
starts at coordinate~$(x,\infty)$ for an input horizontal value~$x$.
Without loss of generality, we will only consider Pachinkos with
balls input at coordinate $x=0$ with unit diameter, 
since the dynamics of a Pachinko with a non-unit diameter input ball 
centered somewhere else can be equivalently simulated via a scaling
and translation. A \emph{Pachinko} is then the set of pins in the environment 
together with the ball's input position.
Given a pin $p$, we will notate its horizontal location as $p_x$ and its
height as $p_y$. 

The ball falls continuously straight down until obstructed by a pin;
see the left of Figure~\ref{fig:trajectory}.
If the bottom half of the ball touches a pin
whose $x$ coordinate is (left, right) of the ball center,
then the ball rotates (clockwise, counterclockwise) around the pin
until the pin reaches the (leftmost, rightmost) point of the ball, no longer
obstructing the ball's downward trajectory.
If the $x$ coordinates match, then the flip of a fair coin
determines the roll direction.

\begin{figure}[htbp]
\centering
\includegraphics[width=6.5 in]{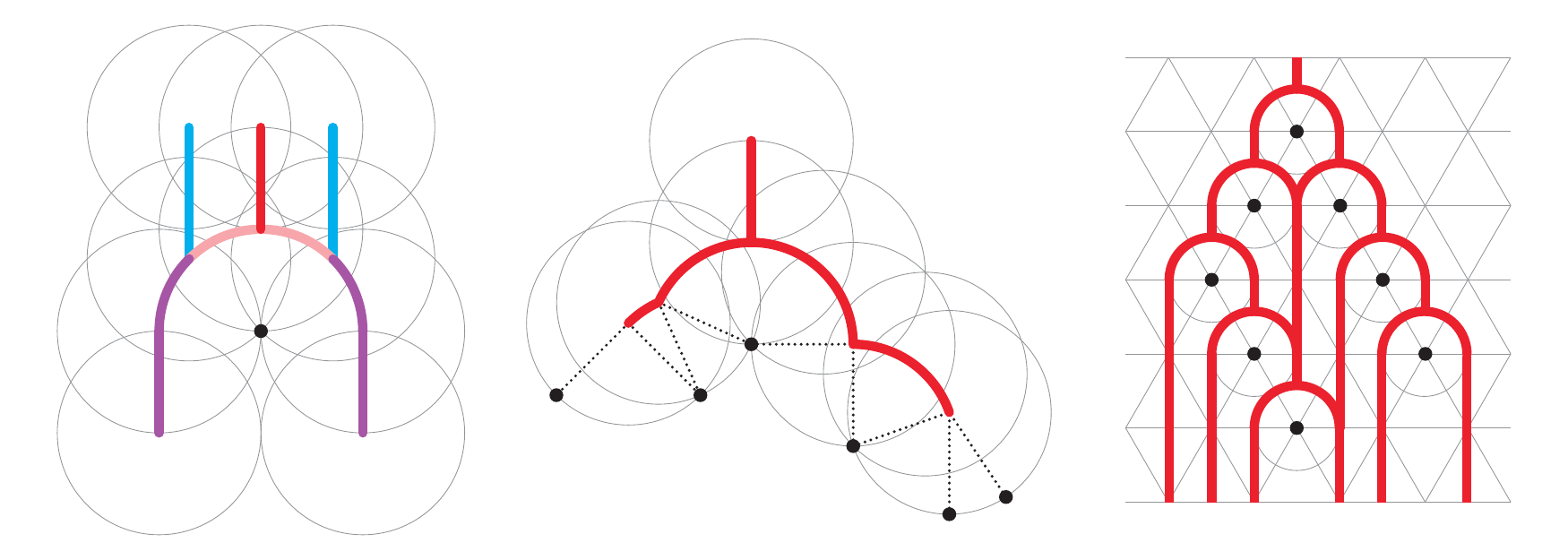}
\caption{Trajectory of falling ball in general model. Outlines of falling
balls are lightly shown, with the path of the ball's center emphasized. Dotted
lines are shown between ball locations that touch more than one pin and 
and the pins touched.
[Left] Balls striking a pin to the left or right of center (blue) rotate to the same side, 
while centered balls (red) fall to either side with equal probability. 
[Center] Balls that touch more than one pin while falling either 
continue rotating around the lower pin,
or get stuck if the center of the ball lies strictly between two touching pins.
[Right] A 50-50 Grid Pachinko.}
\label{fig:trajectory}
\end{figure}

In any case, the ball's rolling around a pin may at some point
be obstructed by another pin or pins. If all obstructing pins 
are either directly below or to one side of the ball center, the ball 
rotates away from the obstructing pins around the obstructing pin 
horizontally closest to the ball's center. 
Alternatively, if obstructing pins exist on 
both sides of the ball center, then the ball stops.
We call the ball \emph{stuck} in this case, 
coming to rest at a \emph{rest site} $s_i$, the location
of the bottom of the ball at rest. This condition is illustrated
in the center of Figure~\ref{fig:trajectory}.
The other possibility for termination is that the ball
reaches $(x'_i,-\infty)$ for some output value~$x'_i$. 
We call the ball \emph{dropped} in this case, with 
$(x'_i,-\infty)$ a \emph{drop site}. 
We call rest and drop sites collectively as
\emph{outputs} to the Pachinko. The number of outputs
is trivially bounded from above by twice the number of pins.

A \emph{50-50 Pachinko} is a Pachinko with the property that 
whenever a ball hits a pin, it hits dead on so that the ball rotates around the pin
to either side with equal probability.  This model is equivalent to requiring
that the pins lie horizontally at integer values away from zero, the initial $x$
coordinate of the ball, such that the ball can never hit two pins at the same time.
But since decreasing the vertical distance between pins of such a Pachinko while
maintaining that the ball only hits one pin at a time will
not change the output dynamics, without loss of generality we can compress pins
to lie on a unit equilateral triangular grid, see the right of Figure~\ref{fig:trajectory}. 
This observation motivates the study of Grid Pachinkos in Section~\ref{sec:grid}.

\subsection{Pachinko Graph}

It would be convenient to abstract away the geometry of a Pachinko and 
 represent the Pachinko as an ``equivalent" graph that is easier to analyze.

Consider the augmented directed graph with vertices corresponding to the pins, 
the ball input, and the outputs of a Pachinko. 
We store with each vertex with a location: the location
of the pin for pins, $(0,\infty)$ for the ball input, $(x',-\infty)$ for drop sites, 
and the location of the bottom of a stuck ball for rest sites. As mentioned previously,
the number of outputs is at most linear in the number of pins, so the number of
vertices in a Pachinko graph is also at most linear.
We add arcs to the directed graph in the following three ways.
\begin{enumerate}[(a)]
\item We add an arc from pin or input $p$ to pin or 
drop site $q$ if an input ball can hit or drop to $q$ directly 
after and not at the same time as hitting $p$, without
getting stuck or also hitting another pin. 
\item If a ball hitting $q$ directly after hitting $p$ gets stuck,
add an arc from both $p$ and $q$ to a vertex corresponding to the rest site $s$.
\item Lastly, if a ball hits more than one vertex directly after hitting $p$ without
getting stuck, add arcs in a chain from $p$ to the highest incident vertex, and from
the highest incident vertex to the next highest and so on.
\end{enumerate}

\begin{figure}[htbp]
\centering
\includegraphics[width=6.5 in]{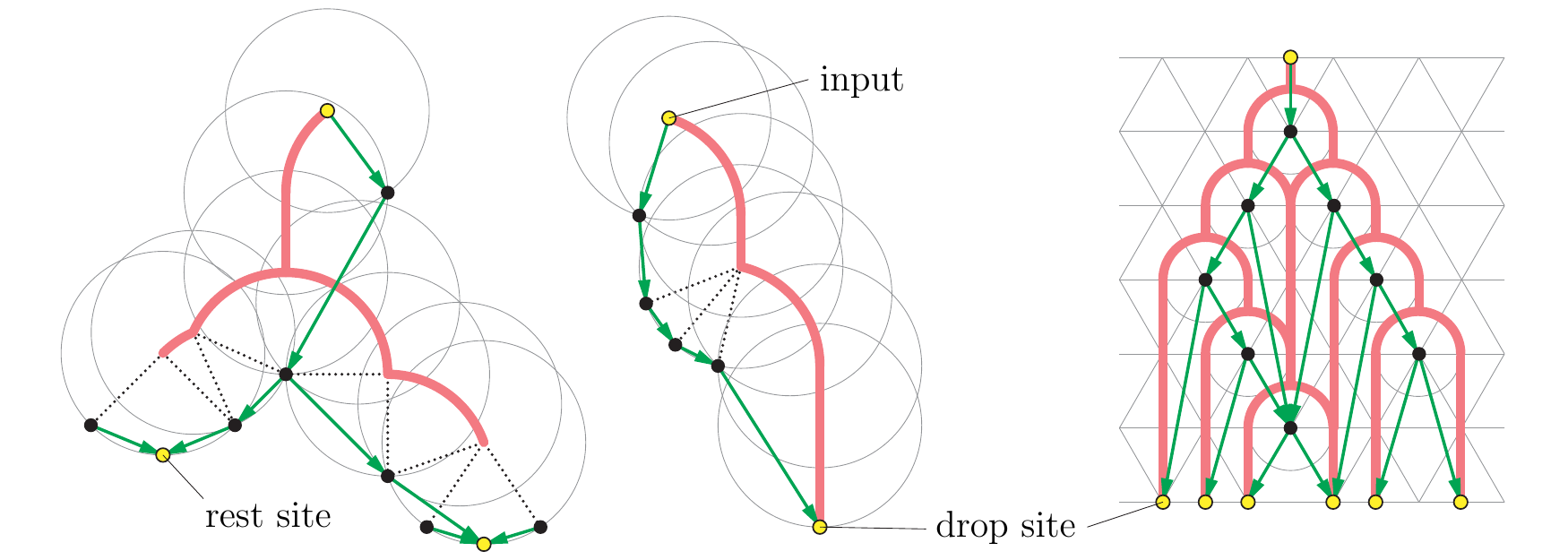}
\caption{Examples of Pachinko graphs shown in green. 
Inputs and outputs are shown as yellow circles. 
[Left] A Pachinko graph containing rest sites. 
[Center] A Pachinko graph hitting more than one
pin at a time without getting stuck.
[Right] A Pachinko graph of a 50-50 Grid Pachinko.}
\label{fig:graph}
\end{figure}

Figure~\ref{fig:graph} illustrates these cases. 
We call the resulting augmented directed graph a \emph{Pachinko graph}. 
For a Pachinko containing $n$ pins, we will show how to 
construct this graph in $O(n\log n)$ time, but first we will analyze some of the 
properties of Pachinko graphs.

\begin{lemma}
\label{lem:geom}
For every arc in a Pachinko graph, the pin corresponding to the tail 
is strictly higher than the pin or output corresponding to the head 
unless they are at the same location. The horizontal
distance traversed by any arc must be less than one ball width.
\end{lemma}
\begin{proof}
To prove the first claim, assume for contradiction there 
exists arc $(p,q)$ with $p_y \leq q_y$ with $p\neq q$. 
If $q$ is not a rest site, an input ball hitting $q$ 
directly after hitting $p$ would rotate around the top of $p$
in the direction of $q$ and then necessarily get stuck between them, a contradiction.
Alternatively, $q$ is a rest site with $p$ on the circle with bottom most point at $q$.
But since $p\neq q$, that means $p_y > q_y$ a contradiction.

To prove the second claim, after rotating around any pin a ball drops directly down with the
pin on its side. The ball cannot move horizontally without 
interacting with a pin, so a ball cannot hit a pin
more than a ball width away from the starting pin without hitting another pin first.
\end{proof}

\begin{lemma}
The straight-line embedding of a Pachinko graph is acyclic and planar, 
with out-degree at most two at any vertex.
\end{lemma}

\begin{proof}
To prove acyclicity, assume for contradiction that a directed cycle exists. The cycle
cannot contain outputs as vertices corresponding to drop and rest sites
have no children. By Lemma~\ref{lem:geom}, the parent of every vertex 
then corresponds to a higher pin, so there can be no highest one, a contradiction. 

To prove planarity, assume for contradiction that two arcs $(a,b)$ and $(c,d)$ 
cross with $a_y>b_y$ and $c_y>d_y$; see Figure~\ref{fig:cross}. 
Without loss of generality, assume $a_y\geq c_y > b_y$, and 
$|a_x-c_x| \geq 1$ or else a ball from $a$ would hit $c$ before hitting 
or outputting at $b$. 
By Lemma~\ref{lem:geom}, $|a_x-b_x| < 1$ and $|c_x-d_x| < 1$, so both
$b$ and $d$ are horizontally between $a$ and $c$. Since the arcs cross,
$d$ must be horizontally between $a$ and $b$, and $b$ must be horizontally
between $c$ and $d$.
Then $b$ cannot be above $d$ or else a ball from $c$ would
hit or output at $b$ before $d$, and $d$ cannot be above $b$ or else a ball from $a$
would hit or output $d$ before $b$, a contradiction.

\begin{figure}[htbp]
\centering
\includegraphics[width=3 in]{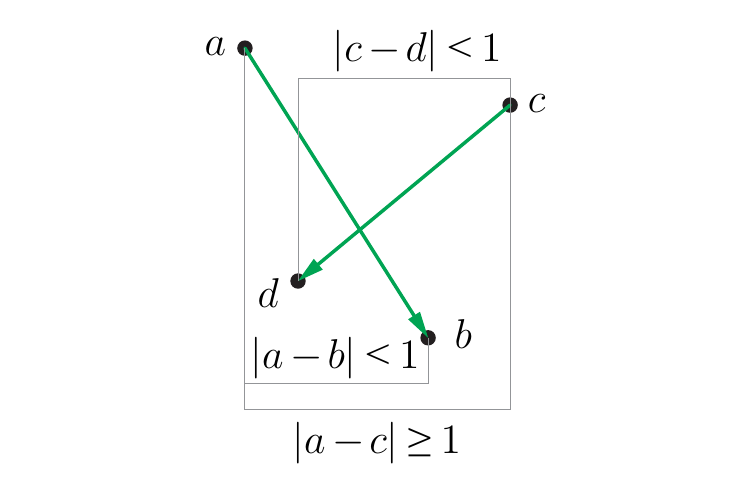}
\caption{Contradicting crossings in a Pachinko graph.}
\label{fig:cross}
\end{figure}

To prove every vertex has max out-degree at most two, first note that every 
output has out-degree zero. Further, arcs can leave pin location $p$
only when a ball hitting $p$ hits another pin or outputs directly after rolling
to the left or the right. Suppose for contradiction that more than two arcs leave $p$,
so two arcs to pins/outputs $q$ and $r$ exist to one side of $p$. 
So $q$ and $r$ are reached directly after rolling around $p$ to one side,
and are reached at the same time.
These arcs cannot be constructed by construction method: (a) because $q$
and $r$ are reached at the same time; (b) because arcs would only be added to
the rest site; nor (c) because an arc would only be added to the higher of $q$
or $r$. So arcs to $q$ and $r$ cannot exist, a contradiction.
\end{proof}

Since a Pachinko graph with $n$ pins has a linear number of vertices and 
bounded out-degree, it follows directly that Pachinko graphs have a 
linear number of edges; so storing a Pachinko graph along with 
vertex locations requires $\Theta(n)$ space.

%

\begin{theorem}
\label{thm:construct}
A Pachinko graph can be constructed from a Pachinko with $n$ pins
in $O(n\log n)$ time.
\end{theorem}

To prove this Theorem, we will construct the Pachinko graph by inductively constructing 
increasing subsets of the Pachinko graph, 
each containing only the arcs terminating at one of the $k$ highest pins. 
At each step, we maintain a Pachinko graph subset, and a 
sorted list of action sites ordered by horizontal position.

First we sort the pins by height to schedule which pin to add next, breaking ties
left before right. 
Next, we construct the Voronoi diagram~\cite{de2000computational} of the Pachinko pins, which has linear size. 
Each Voronoi cell $V_i$ contains one pin $p_i$. Consider the highest point $h_i$ in 
$V_i$ also contained in the unit diameter disk centered at $p_i$. If $h_i$ is not on
the unit diameter semicircle above and centered at $p_i$, 
then $p_i$ can never be reached by the input ball since pins in adjacent
Voronoi cells would block any ball from touching $p_i$. 
Alternatively if $h_i$ is on the semicircle, 
for each pin store the left and right endpoints $\ell_i$ and $r_i$ 
of the largest arc of the semicircle completely inside $V_i$ containing $h$. 

\begin{lemma}
Any ball hitting $p_i$ will have its center on the largest semicircle arc in $V_i$ containing $h_i$. 
\end{lemma}
\begin{proof}
Suppose for contradiction a ball hitting $p_i$ does not
have its center on this arc, and instead hits some point $h'$ on some other disconnected
arc of the semicircle. Then at least one Voronoi edge to some higher pin $p_j$ 
lies between the two disconnected arcs, with $p_i$ between $p_j$ and $h_i$ since $h_i$ is
at least as high as $h'_i$.  But then $p_j$ is above $h'_i$ blocking any ball from reaching $h'_i$,
a contradiction.
\end{proof}

A ball hitting $p_i$ may proceed to roll around it until it reaches $\ell_i$ or $r_i$, 
at which point it will stop or leave contact with $p_i$. We call the input site together with the set of 
$\ell_i$ and $r_i$ for each pin \emph{action sites}. The action sites will serve as 
infrastructure to construct the arcs between pins, the input site, and the output sites (drop sites
and rest sites) that form the Pachinko graph. When we construct an arc of the Pachinko graph
from an action site, what we really mean is to construct an arc 
from the input site if the action site is the input site, 
or from pin $p_i$ if the action site is $\ell_i$ or $r_i$. We say that $\ell_i$ and $r_i$
correspond to pin $p_i$.

We label each endpoint as either DROP or REST depending on if 
a ball reaching the endpoint will continue to move or not.
Endpoints always lie on the semicircle above and centered on $p_i$
by definition. If the endpoint does not also intersect a Voronoi edge, than
it is the endpoint of the semi circle. The ball will then fall away from $p_i$, so we label it DROP. 
Otherwise, the endpoint lies on a Voronoi edge or vertex, and is exactly  a half unit distance
from some subset of pins $Q\ni p_i$, all at or below $p_i$ or else the pin above would block
a ball from ever reaching the endpoint.
If $(\ell_i,r_i)$ is to the (left, right) of $p_i$ and
all the pins of $Q$ are below or to the (right, left) of $(\ell_i,r_i)$, 
then a ball at the endpoint
may continue to move, and we label it DROP. 
Alternatively if $(\ell_i,r_i)$ is to the (left, right) of $p_i$ and not all the pins of $Q$ 
are below or to the (left, right) of $(\ell_i,r_i)$, then a ball at the endpoint
will stop by definition, so we label it REST. 
Lastly, if $(\ell_i,r_i)$ is above or to the (right, left) of $p_i$, then a ball
at the endpoint will rotate around $p_i$ to the other endpoint, so we label it DROP.
These cases are illustrated in Figure~\ref{fig:endpoints}.

\begin{figure}[htbp]
\centering
\includegraphics[width=6.5 in]{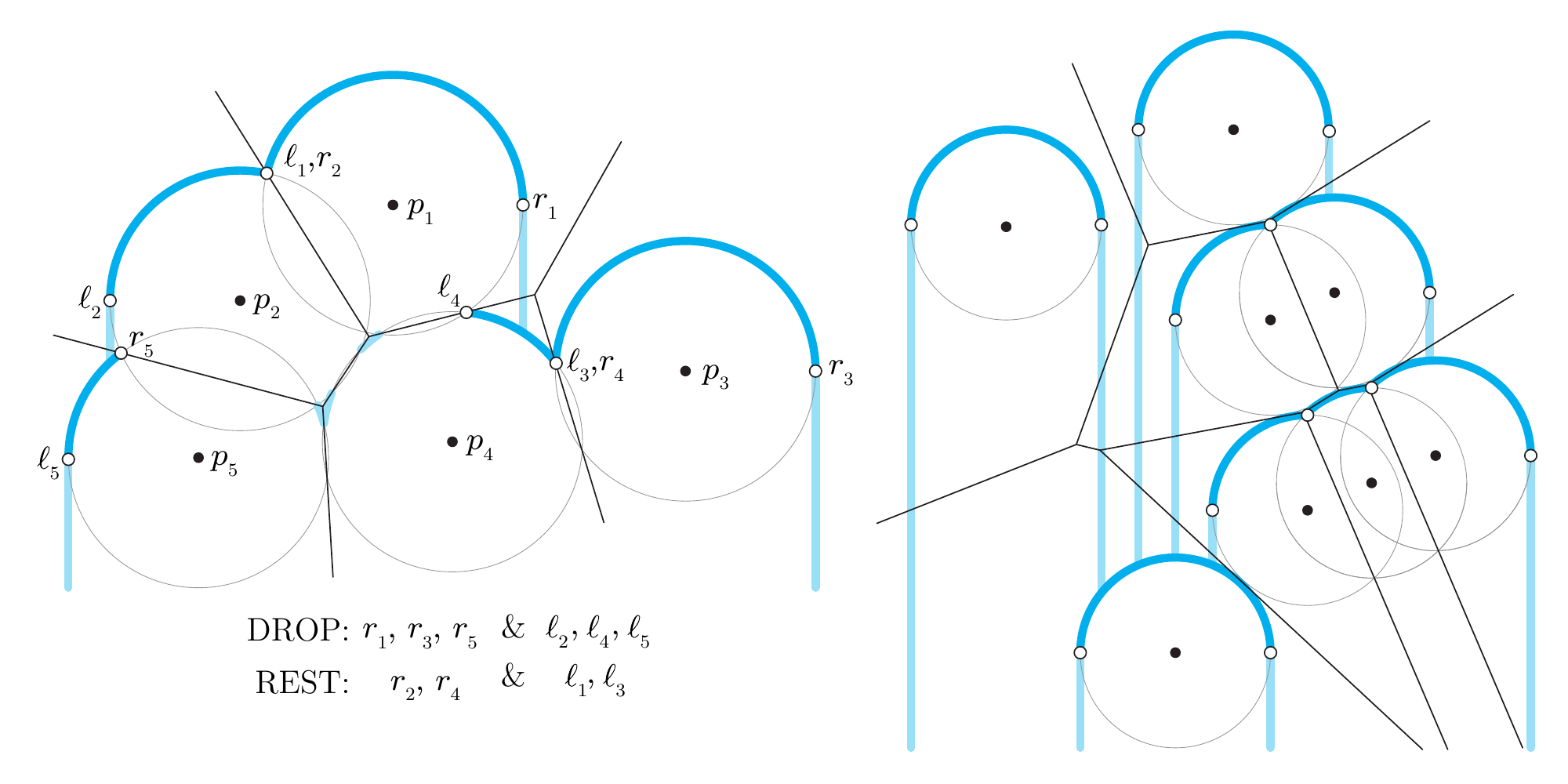}
\caption{Examples of two Pachinkos highlighting the hittable arc for each pin, including their left and
right output endpoints in different environments. Balls that reach endpoints labeled DROP and REST 
will continue to move or not respectively.}
\label{fig:endpoints}
\end{figure}



For the base case, there are no existing pins or action sites to consider, 
and the Pachinko graph subset contains no arcs as desired. 
We will maintain a list of \emph{active} action sites.
Since the ball drops from the input location, we initialize the list 
of active action sites with the ball's input location $(x,\infty)$ as a DROP action site.

For the inductive case, we are given the subset of the Pachinko graph containing only
the arcs terminating at one of the $k$ highest pins and a horizontally sorted list of active 
sites corresponding to all locations of balls that may leave contact with the $k$ highest pins.
Let $p_i$ be the $k+1$ highest pin, breaking ties left before right, and let $V_i$ be its Voronoi cell. 
Binary search for the set $A_i$ of all DROP action sites horizontally within half a unit of $p_i$.
Now update the Pachinko graph for each DROP action site $a\in A_i$ by
adding to the Pachinko graph subset an arc from $a$ to $p_i$. To add new active action sites, 
if there is any action site in $A_i$ to the (left,right)
or above $p_i$, add $(\ell_i,r_i)$ to the sorted active site list, while removing all DROP 
action sites in $A_i$ from the active site list. 

\begin{lemma}
The above procedure constructs the Pachinko graph containing only the arcs 
terminating at one of the $k+1$ highest pins and a horizontally sorted list of active sites
corresponding to all locations of ball centers that can leave the $k+1$ highest pins. 
\end{lemma}

\begin{proof}
To prove the first claim, we need only check that the added arcs are exactly the arcs terminating
at $p_i$. By construction, all arcs end at $p_i$. Now suppose for contradiction that some arc
of the Pachinko graph terminating at $p_i$ from pin or the input is not added by the above procedure. 
If the arc starts from a pin $p_j$, then it is one of the $k$ highest pins by Lemma~\ref{lem:geom}, 
and either $\ell_j$ or $r_j$ will be labeled DROP or else a ball hitting $p_j$ would not leave. 
Further, either $\ell_j$ or $r_j$ (or the input location) must be within a horizontal half unit 
distance from $p_i$ by Lemma~\ref{lem:geom}, with no pin obstructing the path between them.
But then $p_i$ would have a corresponding DROP action site above it, and the above procedure
would construct the corresponding arc to $p_i$, a contradiction. 

To prove the second claim, all action sites removed by the procedure certainly 
would not be able to reach any pin below $p_i$. Further, $\ell_i$ and $r_i$ correspond
to the only locations a ball could leave $p_i$, proving the claim.
\end{proof}

After constructing all the arcs of the Pachinko graph terminating at pins, it remains to construct the arcs 
to the outputs. Add a drop site at $(x',-\infty)$ corresponding to the horizontal position of each active DROP
action site and an arc from the action site to the drop site. 
It is possible that multiple REST endpoints coincide.
In any case, we construct a single rest site for all
REST endpoints in the same location, 
with location half a unit directly below the location, and construct
an arc from each REST action site at the location to the constructed rest site. 
The correctness of these constructions follow directly
from the definitions of drop, rest, and action sites.
An example of this algorithm applied to a Pachinko is shown
in Figure~\ref{fig:algorithm}.

\begin{figure}[htbp]
\centering
\includegraphics[width=6.5 in]{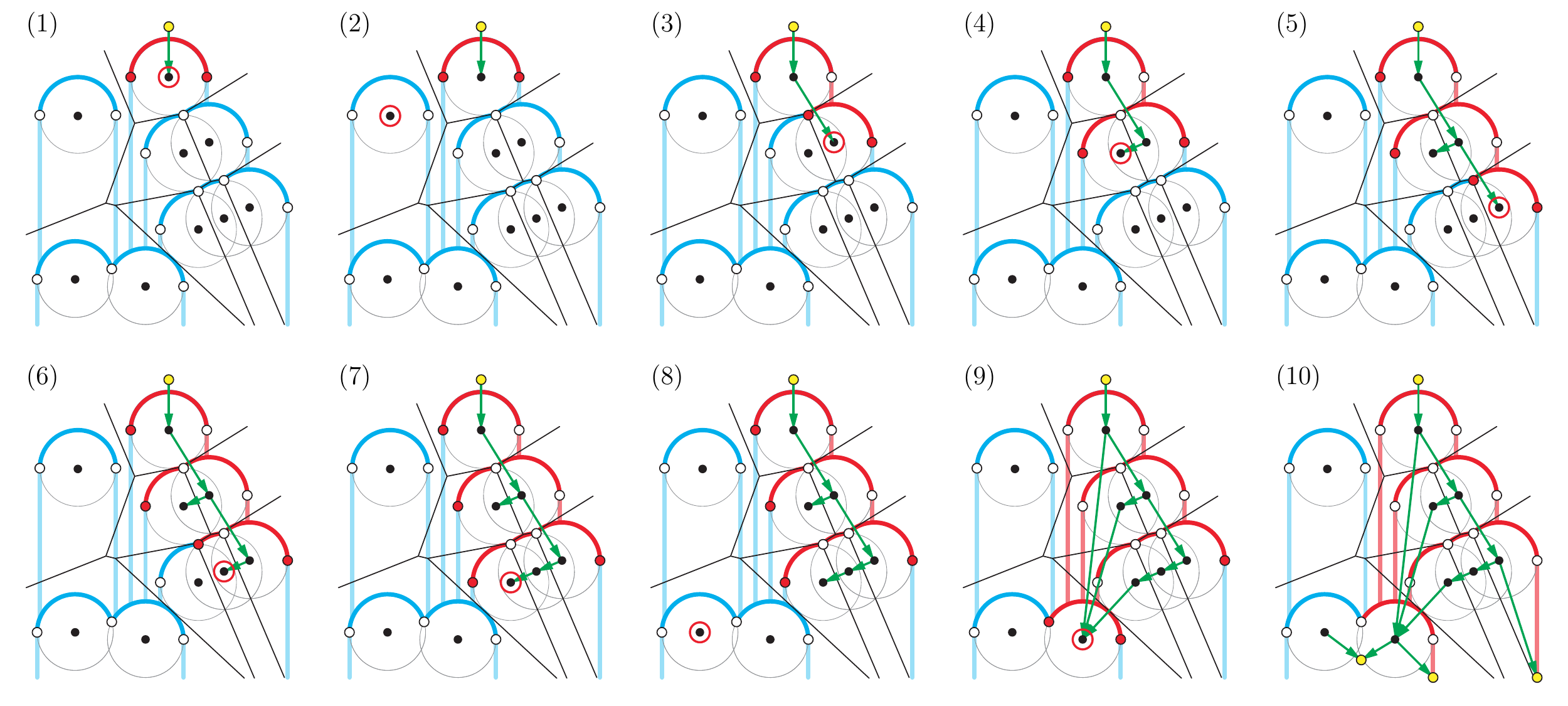}
\caption{Example showing the described Pachinko graph construction algorithm for 
an input Pachinko. Blue paths show possible ball paths for any input, 
while the red paths show ball paths for the given input. The target pin is circled
in red, active sites are colored red, and Pachinko graph arcs are shown in green.
}
\label{fig:algorithm}
\end{figure}

Now to prove Theorem~\ref{thm:construct}. \\

\begin{proof}
Construct the Pachinko graph using the above procedure. Sorting and construction of the Voronoi diagram
each takes $O(n\log n)$ time. The calculating and labeling intersections between circles centered 
at pins and the Voronoi cell containing them is at most linear since the number of edges in the 
Voronoi diagram is linear and a circle can cross a line at most twice, and each can be 
calculated in constant time. The number of action sites is linear since it is bounded above by
twice the number of pins, and binary searching for DROP action sites from the horizontally sorted list
of active action sites with location near a pin takes $O(\log n)$
time per pin. Further, since each DROP action site is immediately removed after being found, the
total search and maintenance time is $O(n\log n)$, leading to an $O(n\log n)$ total construction time. 


%
%


\end{proof}


\begin{theorem}
The probability that an input ball hits any pin of a Pachinko or outputs
at a rest or drop site can be calculated for all pins and outputs in $O(n\log n)$ time. 
\end{theorem}
\begin{proof}
To calculate probabilities, construct the Pachinko graph, and calculate probabilities in the graph 
breadth first. First we calculate the probability transferred
along each Pachinko graph arc. If the edge starts at the input site, 
the edge carries probability 1. Otherwise, the arc $a$ starts at a pin $p$ (with out-degree
at most two, possibly one to the left and one to the right) and points to
a vertex $v$ either on the left or right of $p$, and we calculate the probability transferred
along the arc by summing probabilities coming from the parents of $p$. 
For each arc $a'$ terminating at $p$ starting from a vertex $u$ either more than half a unit distance
horizontally from $p$ on the same side as $v$ or less than half a unit distance horizontally
from $p$ on the opposite side as $v$, transfer the total probability of $a'$ to $a$. 
This assignment is correct because any balls falling from $u$ to $p$ falls with its center
on the same side as $v$, so will always output on that side.
Additionally, for each arc $a'$ terminating at $p$ starting from a vertex $u$ exactly half a unit 
distance horizontally from $p$, transfer half the probability of $a'$ to $a$. 
This assignment is correct because any ball falling from $u$ to $p$ falls centered with $p$
and will split its probability between its two children. 
These cases are illustrated in Figure~\ref{fig:prob}.
The probability that a ball hits any pin $p$ is then the sum of the probabilities of all arcs 
into $p$. All arc probabilities can only contribute to any sum once, 
so all probabilities can be calculated in breadth first order in linear time. Thus
all probabilities can be be calculated in $O(n)$ time on top of $O(n\log n)$ time
needed to construct the Pachinko graph.

\end{proof}

\begin{figure}[htbp]
\centering
\includegraphics[width=6.5 in]{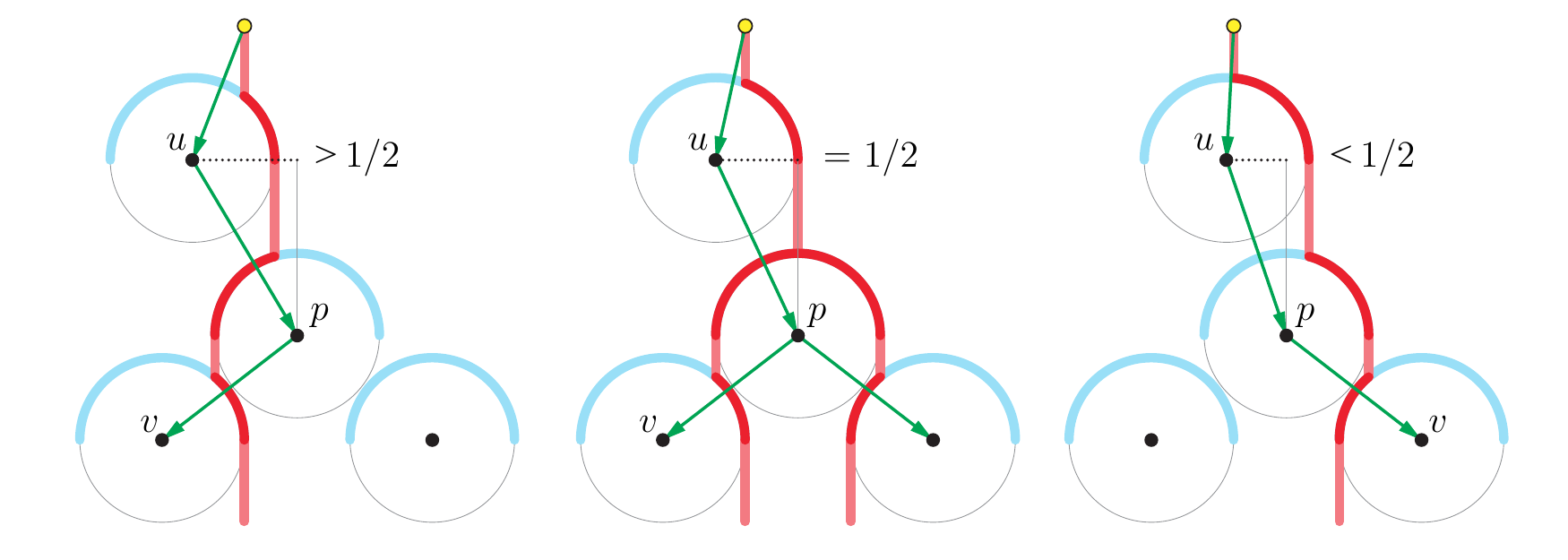}
\caption{Three cases for routing probability through a vertex $p$.
[Left] A vertex $u$ more than half a unit distance
horizontally from $p$ on the same side as $v$. 
[Center] A vertex $u$ exactly half a unit 
distance horizontally from $p$.
[Right] A vertex $u$ 
less than half a unit distance horizontally from $p$ on the opposite side as $v$.
}
\label{fig:prob}
\end{figure}
 
\subsection{Grid Pachinkos}
 \label{sec:grid}
 
General Pachinko can have complicated behavior, especially when three or more 
pins can touch a ball at the same time. Thus in constructing Pachinkos, it will be easier
to restrict ourselves to more idealized Pachinkos, namely Pachinkos with pin locations
and ball input location restricted to lie on a unit equilateral triangular grid. 
Without loss of generality, we let the ball input column be column zero, 
with columns every half unit numbered increasing to the right and 
decreasing to the left. We will call row 1 the row containing the highest pin at height $h$, 
with row $k$ containing pins at height $h-k\sqrt{3}/2$.  
Recall that such a setup results in 50-50
Pachinkos, with each pin causing an incident ball to rotate to the left or right with
equal probability. These pins act normally, so we call them N-pins.

However, to allow Grid Pachinkos to simulate more of the behavior of a general Pachinko,
we allow the placement of L-pins, R-pins, and S-pins at grid vertices, pins where incident balls must roll
left, right, or stop respectively. We can simulate Grid Pachinkos using general
Pachinko dynamics by simulating N-pins, L-pins, R-pins, and S-pins with pin arrangements 
contained in $\delta\times\delta$ blocks as shown in Figure~\ref{fig:grid}, 
and increasing the distance between each row and column by $\delta$. 
Since we can simulate Grid Pachinkos with L-pins, R-pins, and S-pins using general Pachinko 
dynamics, any output distributions we can construct using Grid Pachinkos we can
also construct using general Pachinkos with three times the number of pins. A nice
property of Grid Pachinkos is that the probability flow through the Pachinko graph
is independent of the location of pins, and computable simply from the structure of the
Pachinko graph.

\begin{figure}[htbp]
\centering
\includegraphics[width=6.5 in]{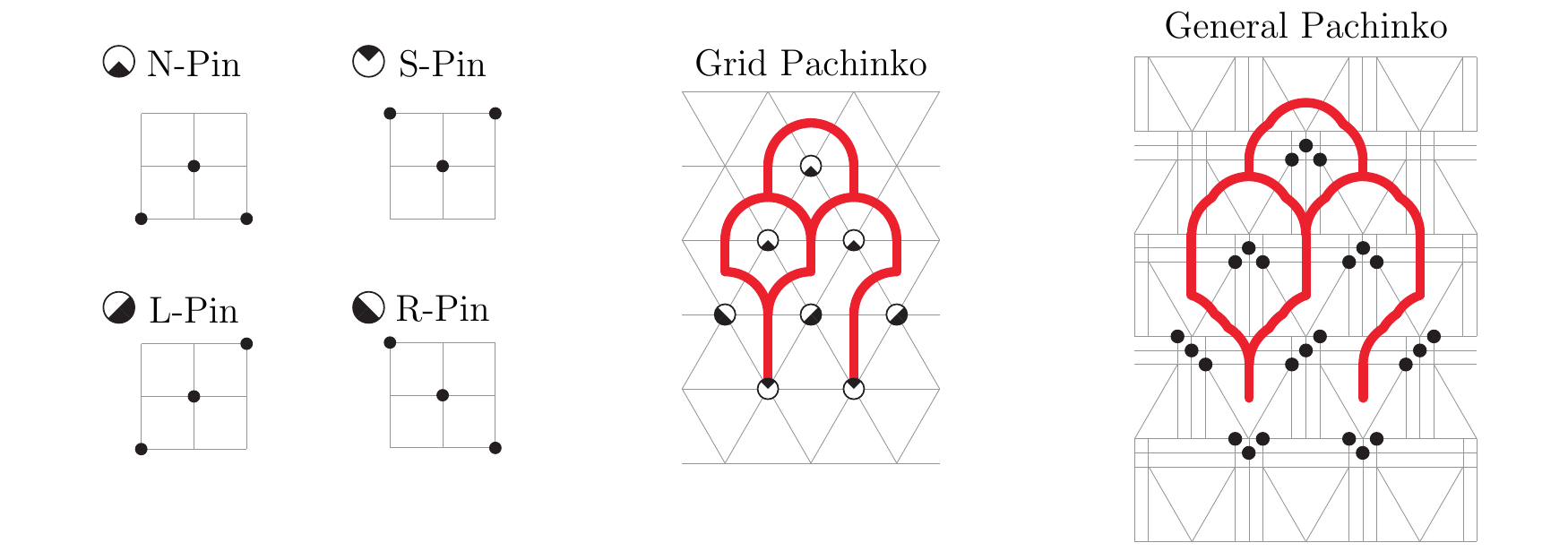}
\caption{Simulating Grid Pachinkos with General Pachinkos. 
Each pin gadget is $\delta \leq 1/2$ tall and wide.
Modifying the triangular grid allows space for the gadgets.}
\label{fig:grid}
\end{figure}

\section{50-50 Pachinkos}

50-50 Pachinkos are a subset of general Pachinkos, and are more restrictive
in the possible distributions one can construct. 

\subsection{Invariants on the Distribution}

Unlike General Pachinkos or Grid Pachinkos that
can use L-pins and R-pins to move probability arbitrarily to any location, 50-50
Pachinkos must output all probability to drop sites in a ``centered" manner.

\begin{lemma}
\label{lem:span}
If there are $k > 1$ outputs, they span strictly between $k$ and $2k$ columns with no two consecutive columns lacking an output.
\end{lemma}
\begin{proof}
There must be at least $k+1$ columns spanned by the outputs, $k$ for the outputs, and one column necessarily blocked by the lowest hit pin. No two consecutive columns lack an output. Suppose for contradiction two such columns exist for which there exist outputs to the left and right of these columns. Then there must exist hit pins in both columns or else a ball could not travel across them. The lowest such pin in each column cannot be next to each other because of the grid, so one must be above the other. The lower pin will necessarily output in the column of the higher pin, a contradiction.
\end{proof}

\begin{theorem}
\label{thm:centered}
If every pin of a 50-50 Pachinko is between columns $-t$ and $t$, and $p_i$ is the probability that the ball outputs in column $i$, then $\sum_i i p_i = 0$.  
\end{theorem}

\begin{proof}
The proof is by induction on the number of pins. For 0 pins, it's clearly true. Given an $n$-pin arrangement, remove a bottommost pin to get an arrangement for which, by induction, $\sum_i i p_i = 0$. Adding a pin at the bottom of column $k$ replaces probability $p_k$ in column $k$ with probability $\frac{p_k}2$ in each of columns $k-1$ and $k+1$, changing the total by $\frac{p_k}{2}(k-1) - p_k k + \frac{p_k}{2}(k+1) = 0$, as desired.
\end{proof}

However, not every dyadic probability distribution satisfying the above condition is the set of probabilities of a 50-50 Pachinko. For instance, the two output probability distribution $\langle1/4, 0, 3/4\rangle$ is not constructible by a 50-50 Pachinko as a consequence of Theorem~\ref{thm:centered}, but it is constructible by a General or Grid Pachinko, as shown in the left diagram of Figure~\ref{fig:threefour}. It is still open whether every dyadic probability can be output by a 50-50 Pachinko. 

\begin{figure}[htbp]
\centering
\includegraphics[width=6.5 in]{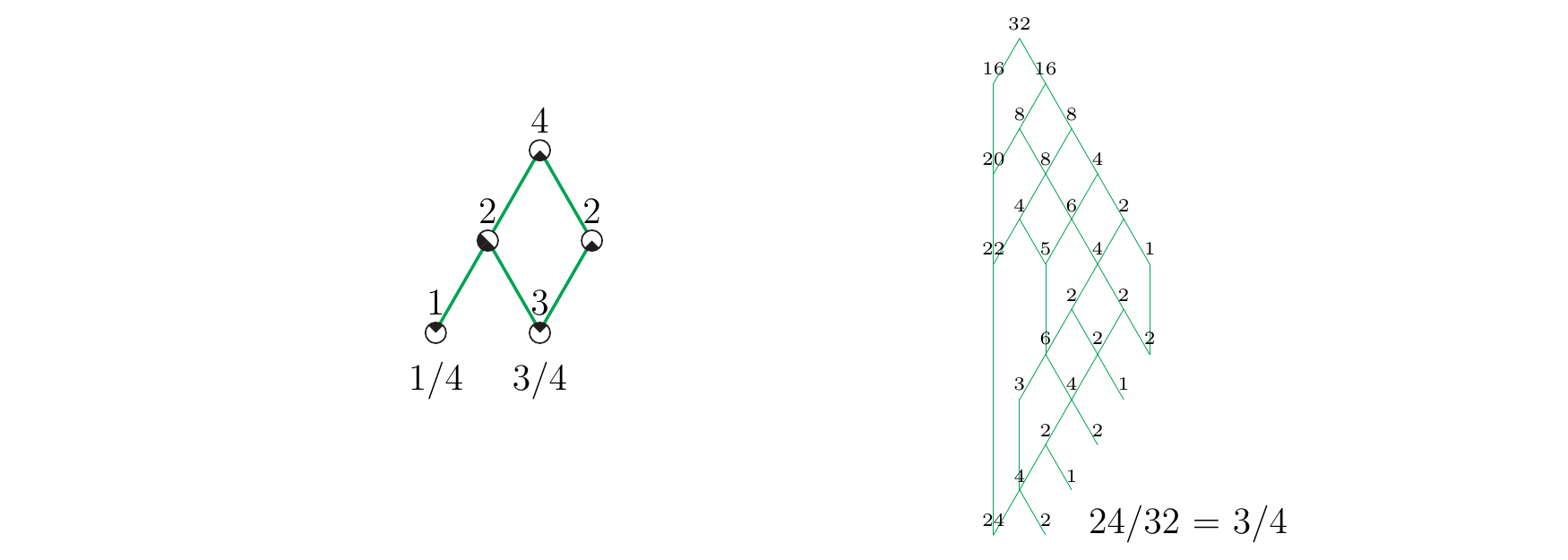}
\caption{[Left] A Grid Pachinko constructing the distribution $\langle1/4, 0, 3/4\rangle$. [Right] Constructing output probability $3/4$ with a 50-50 Pachinko.}
\label{fig:threefour}
\end{figure}

\subsection{Constructible Probabilities}

\begin{open}
  Is every dyadic rational the output probability of some 50-50 Pachinko?
\end{open}

Table~\ref{tab:exhaust} gives all dyadic probabilities
that can be output by all possible 50-50 Pachinkos with a small number of rows. 
This table was constructed by naive exponential exhaustive search. 
This data suggests that some probabilities require a large number of rows of pins to produce. 
For example, the probability $3/4$ (0.11 in binary) cannot be produced by a 50-50 Pachinko 
with fewer than nine rows of pins, but it can be produced by a 50-50 Pachinko with eleven rows; 
see the right diagram in Figure~\ref{fig:threefour}.

\begin{table}
  \centering
  \tiny
  \begin{tabular}{c|p{3in}|p{3in}}
    $n$ & Computable & Uncomputable \\ \hline
$1$ & 0, .1, 1. & none \\ \hline
$2$ & 0, .01, .1, 1. & .11 \\ \hline
$3$ & 0, .001, .01, .011, .1, .101, 1. & .11, .111 \\ \hline
$4$ & 0, .0001, .001, .0011, .01, .0101, .011, .0111, .1, .1001, .101, 1.
    & .1011, .11, .1101, .111, .1111 \\ \hline
$5$ & 0, .00001, .0001, .00011, .001, .00101, .0011, .00111, .01, .01001, .0101, .01011, .011, .01101, .0111, .01111, .1, .1001, .101, .10101, .1011, 1.
    & .10001, .10011, .10111, .11, .11001, .1101, .11011, .111, .11101, .1111, .11111 \\ \hline
$6$ & 0, .000001, .00001, .000011, .0001, .000101, .00011, .000111, .001, .001001, .00101, .001011, .0011, .001101, .00111, .001111, .01, .010001, .01001, .010011, .0101, .010101, .01011, .010111, .011, .011001, .01101, .011011, .0111, .011101, .01111, .011111, .1, .100001, .10001, .100011, .1001, .100101, .10011, .101, .101001, .10101, .101011, .1011, 1.
    & .100111, .101101, .10111, .101111, .11, .110001, .11001, .110011, .1101, .110101, .11011, .110111, .111, .111001, .11101, .111011, .1111, .111101, .11111, .111111 \\ \hline
$7$ & 0, .0000001, .000001, .0000011, .00001, .0000101, .000011, .0000111, .0001, .0001001, .000101, .0001011, .00011, .0001101, .000111, .0001111, .001, .0010001, .001001, .0010011, .00101, .0010101, .001011, .0010111, .0011, .0011001, .001101, .0011011, .00111, .0011101, .001111, .0011111, .01, .0100001, .010001, .0100011, .01001, .0100101, .010011, .0100111, .0101, .0101001, .010101, .0101011, .01011, .0101101, .010111, .0101111, .011, .0110001, .011001, .0110011, .01101, .0110101, .011011, .0110111, .0111, .0111001, .011101, .0111011, .01111, .0111101, .011111, .0111111, .1, .1000001, .100001, .1000011, .10001, .1000101, .100011, .1000111, .1001, .1001001, .100101, .10011, .101, .101001, .10101, .1010101, .101011, .1011, .1011001, .101101, .1011011, .10111, .1011101, 1.
    & .1001011, .1001101, .100111, .1001111, .1010001, .1010011, .1010111, .101111, .1011111, .11, .1100001, .110001, .1100011, .11001, .1100101, .110011, .1100111, .1101, .1101001, .110101, .1101011, .11011, .1101101, .110111, .1101111, .111, .1110001, .111001, .1110011, .11101, .1110101, .111011, .1110111, .1111, .1111001, .111101, .1111011, .11111, .1111101, .111111, .1111111 \\ \hline
$8$ & 0, .00000001, .0000001, .00000011, .000001, .00000101, .0000011, .00000111, .00001, .00001001, .0000101, .00001011, .000011, .00001101, .0000111, .00001111, .0001, .00010001, .0001001, .00010011, .000101, .00010101, .0001011, .00010111, .00011, .00011001, .0001101, .00011011, .000111, .00011101, .0001111, .00011111, .001, .00100001, .0010001, .00100011, .001001, .00100101, .0010011, .00100111, .00101, .00101001, .0010101, .00101011, .001011, .00101101, .0010111, .00101111, .0011, .00110001, .0011001, .00110011, .001101, .00110101, .0011011, .00110111, .00111, .00111001, .0011101, .00111011, .001111, .00111101, .0011111, .00111111, .01, .01000001, .0100001, .01000011, .010001, .01000101, .0100011, .01000111, .01001, .01001001, .0100101, .01001011, .010011, .01001101, .0100111, .01001111, .0101, .01010001, .0101001, .01010011, .010101, .01010101, .0101011, .01010111, .01011, .01011001, .0101101, .01011011, .010111, .01011101, .0101111, .01011111, .011, .01100001, .0110001, .01100011, .011001, .01100101, .0110011, .01100111, .01101, .01101001, .0110101, .01101011, .011011, .01101101, .0110111, .01101111, .0111, .01110001, .0111001, .01110011, .011101, .01110101, .0111011, .01110111, .01111, .01111001, .0111101, .01111011, .011111, .01111101, .0111111, .01111111, .1, .10000001, .1000001, .10000011, .100001, .10000101, .1000011, .10000111, .10001, .10001001, .1000101, .10001011, .100011, .10001101, .1000111, .10001111, .1001, .10010001, .1001001, .10010011, .100101, .10010101, .1001011, .10010111, .10011, .10011001, .1001101, .10011011, .100111, .10011101, .1001111, .101, .10100001, .101001, .10100101, .1010011, .10101, .10101001, .1010101, .10101011, .101011, .10101101, .1010111, .10101111, .1011, .10110001, .1011001, .10110011, .101101, .10110101, .1011011, .10110111, .10111, .10111001, .1011101, 1.
    & .10011111, .1010001, .10100011, .10100111, .10111011, .101111, .10111101, .1011111, .10111111, .11, .11000001, .1100001, .11000011, .110001, .11000101, .1100011, .11000111, .11001, .11001001, .1100101, .11001011, .110011, .11001101, .1100111, .11001111, .1101, .11010001, .1101001, .11010011, .110101, .11010101, .1101011, .11010111, .11011, .11011001, .1101101, .11011011, .110111, .11011101, .1101111, .11011111, .111, .11100001, .1110001, .11100011, .111001, .11100101, .1110011, .11100111, .11101, .11101001, .1110101, .11101011, .111011, .11101101, .1110111, .11101111, .1111, .11110001, .1111001, .11110011, .111101, .11110101, .1111011, .11110111, .11111, .11111001, .1111101, .11111011, .111111, .11111101, .1111111, .11111111
  \end{tabular}
  \caption{Computable and uncomputable numbers using all possible
    configurations on $n$ rows, $n \leq 8$.  The list of uncomputable numbers
    restricts to numbers with at most $n$ bits of precision, because beyond
    that nothing is computable with just $n$ rows.}
  \label{tab:exhaust}
\end{table}

\subsection{Full and Truncated Pachinkos}

While not all distributions are constructible, in this and the following sections we analyze some interesting 50-50 Pachinkos. In this section, we consider 50-50 Pachinkos containing pins in all possible locations up to row $k$ (first row with pins is row 1, row number increasing down), which we call the Full $k$-Pachinko. The output probabilities of this Pachinko constitute the $(k+1)$th row of Pascal's triangle divided by $2^k$; see left of Figure~\ref{fig:pascal}. This fact is readily apparent because the probability at each pin is the sum of the two probabilities above it divided by two.

\begin{figure}[htbp]
\centering
\includegraphics[width=6.5 in]{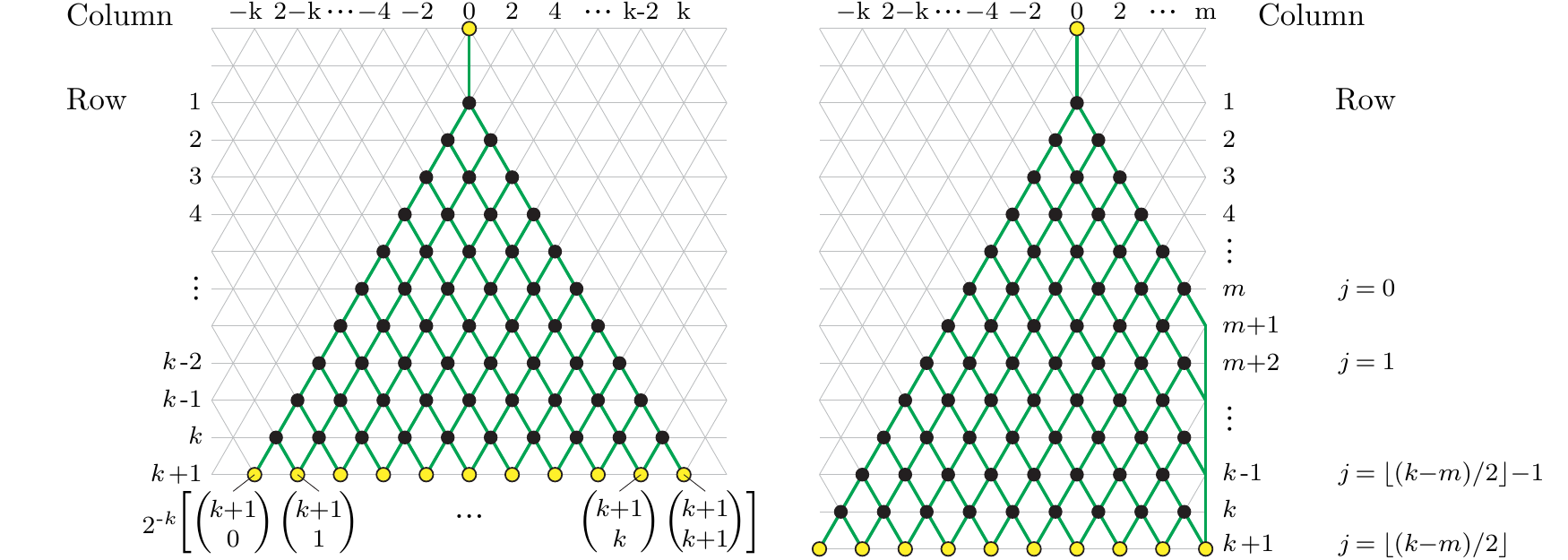}
\caption{A full $k$-Pachinko on left and an $m$-truncated $k$-Pachinko on the right.
The $k$-Pachinko outputs binomial coefficients over powers of 2, while the $m$-truncated
$k$-Pachinko outputs Ballot numbers over powers of 2 down the right side.}
\label{fig:pascal}
\end{figure}

Another interesting family of 50-50 Pachinkos arises by removing from a Full $k$-Pachinko all the pins in column $m$, which we call an $m$-truncated $k$-Pachinko; see right of Figure~\ref{fig:pascal}. Without loss of generality, we will assume $k > m > 0$, the $m$th column to the right of center. A ball may fall into column $m$ only directly after hitting some pin in column $m-1$ in row $m+2j$ for $j\in\{0,\ldots,\lfloor(k-m)/2\rfloor\}$. Let $B(m,j)$ be the number of paths an input ball starting at $(0,1)$ can take to reach pin $(m-1,m+2j)$. Then the probability of a ball falling into column $m$ after hitting pin $(m-1,m+2j)$ is then $B(m,j)/2^{m+2j}$. We observe that every path in the pachinko from $(0,1)$ to $(m-1,m+2j)$ must hit either the pin at $(-1,2)$ or the pin at $(1,2)$, leading to the recurrence $B(m,j) = B(m-1,j) + B(m+1,j-1)$; see Figure~\ref{fig:recurse}. The numbers that satisfy this recurrence are the Ballot numbers~\cite{Graham}, with:

\begin{equation}
B(m,j) = \frac{m}{2j+m}{2j+m\choose j}.
\end{equation}

\begin{figure}[htbp]
\centering
\includegraphics[width=6.5 in]{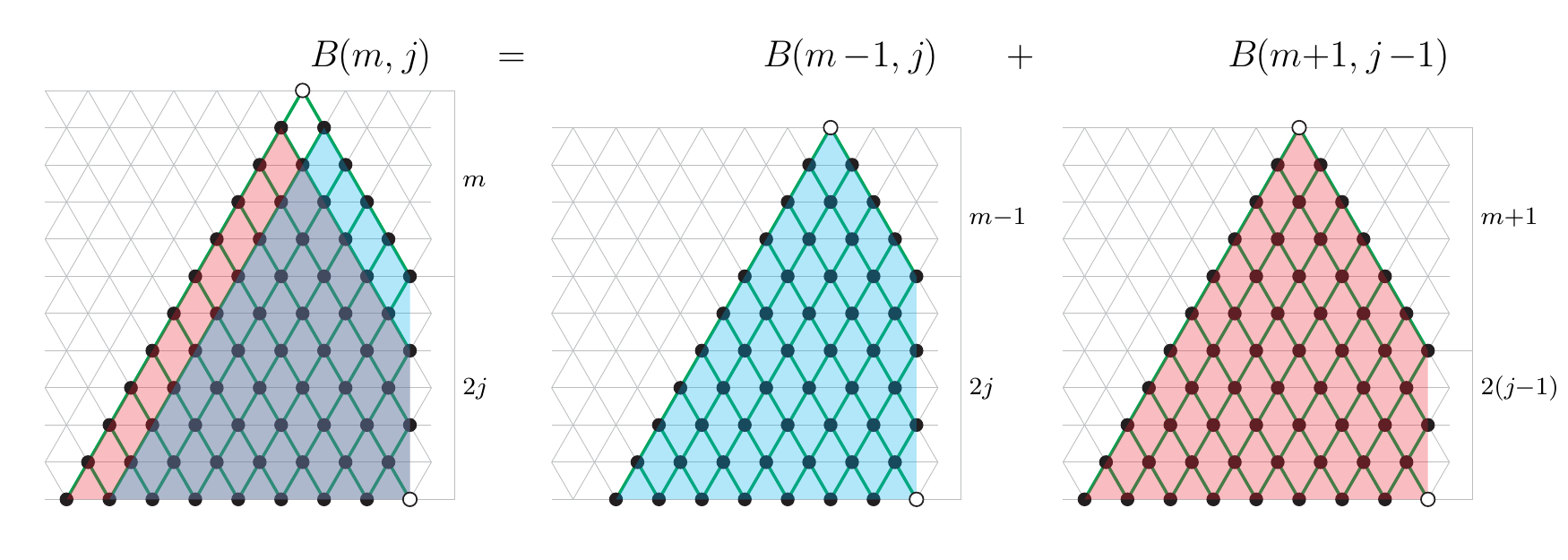}
\caption{Ballot Number recurrence showing that $B(m,j) = B(m-1,j) + B(m+1,j-1)$.}
\label{fig:recurse}
\end{figure}

Ballot numbers are a generalization of the Catalan numbers, the special case for $m = 1$. Ballot numbers have the following generating function~\cite{Graham}:

\begin{equation}
\left(\frac{1-\sqrt{1-4x}}{2x}\right)^m = \sum\limits_{j=0}^\infty B(m,j)x^j.
\end{equation} 

We can use the properties of these numbers to prove two useful Lemmas which we use in Section~\ref{sec:general} to construct Pachinkos that can approximate any probability distribution.

\begin{lemma}
\label{lem:one}
The output probability in column $m$ of an $m$-truncated $k$-Pachinko approaches 1 as $k\rightarrow\infty$.
\end{lemma}
\begin{proof}
\vspace{-1pc}
\begin{equation}
\sum\limits_{j=0}^\infty \frac{B(m,j)}{2^{m+2j}} = \frac{1}{2^m}\sum\limits_{j=0}^\infty B(m,j)\frac{1}{4^j} = \frac{1}{2^m}\left(\frac{1-\sqrt{1-4(1/4)}}{2(1/4)}\right)^m = 1
\end{equation}
\end{proof}
\begin{lemma}
\label{lem:exists}
There exists a finite Pachinko that outputs a ball in column $m$ with probability greater than $1-\varepsilon$ for any $\varepsilon > 0$. 
\end{lemma}
\begin{proof}
Consider an $m$-truncated $(m+2j)$-Pachinko for $j$ such that $B(m,j+1)/2^{m+2(j+1)}~\leq~\varepsilon$, which exists because $B(m,j)$ is always positive and $\lim_{j \rightarrow\infty} B(m,j)/4^j = 0$ (easily derived from Stirling's approximation for factorials~\cite{abramowitz1966handbook}). Then Lemma~\ref{lem:one} yields the following lower bound on the output probability in column $m$:
\begin{equation}
\sum\limits_{i = 0}^{j}\frac{B(m,i)}{2^{m+2i}} = \sum\limits_{i = 0}^{\infty}\frac{B(m,i)}{2^{m+2i}} - \sum\limits_{i = j+1}^{\infty}\frac{B(m,i)}{2^{m+2i}} > 1 - \frac{B(m,j+1)}{2^{m+2(j+1)}}\geq 1-\varepsilon.
\end{equation} 
\end{proof}

\subsection{3 or Fewer Outputs}

For a small number of outputs, constructible output probability distributions of 50-50 Pachinkos can be calculated directly. The only possible outputs for 50-50 Pachinkos with one or two outputs are $\langle1\rangle$ and $\langle1/2,0,1/2\rangle$ respectively. 50-50 Pachinkos with exactly three outputs can output in either four or five columns by Lemma~\ref{lem:span}. For four columns, pins only exist in two columns, and hit pins ordered by height must alternate between the two columns as shown on the left of Figure~\ref{fig:few}. Such a Pachinko with $k$ hit pins results in the following distributions or their reversals:

\begin{align}
\left\langle \sum\limits_{i=1}^{j+1} 2^{-2i+1},0,2^{-k},\sum\limits_{i=1}^j 2^{-2i}\right\rangle&\quad\textnormal{for }k = 2j+1, \\
\left\langle \sum\limits_{i=1}^j 2^{-2i+1},2^{-k},0,\sum\limits_{i=1}^j 2^{-2i}\right\rangle&\quad\textnormal{for }k = 2j.
\end{align}

 50-50 Pachinkos with exactly three outputs in five columns can be represented implicitly by products of linear transformations on a 5-vector encoding the output probability in each column. The ball input is represented by a starting vector chosen from the set
 \begin{equation}
 \mathcal{X} = \left\{\langle 1/2,0,1/2,0,0\rangle^T, \langle 0,0,1,0,0\rangle^T, \langle 0,0,1/2,0,1/2\rangle^T\right\},
 \end{equation} 
 which correspond to a ball starting in column two, three, or four respectively. Then the ball may interact with a sequence of one or more pin arrangements in any of the three patterns shown in the center of Figure~\ref{fig:few}. 
 
\begin{figure}[htbp]
\centering
\includegraphics[width=6.5 in]{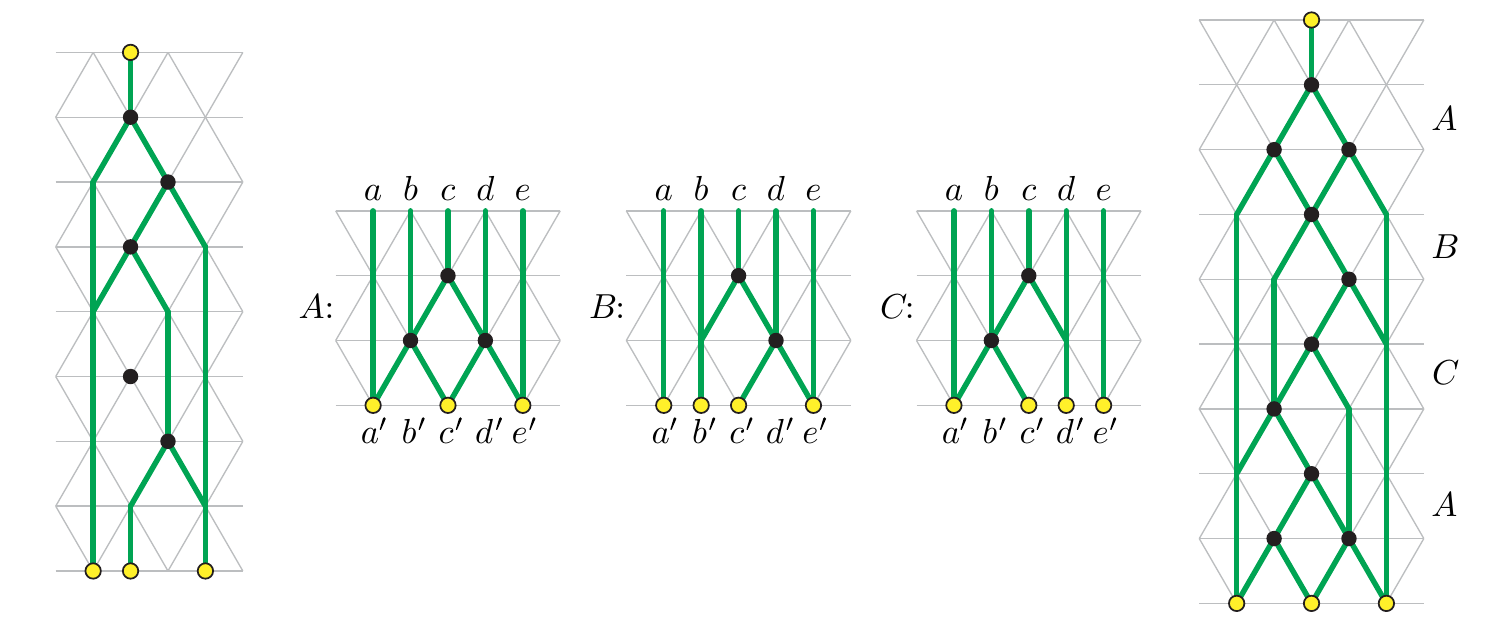}
\caption{[Left] Example of a four column Pachinko with three nonzero outputs. 
[Center] The three pin arrangement patterns for creating five column Pachinkos
with three nonzero outputs. 
For pattern $A$, $(a',b',c',d',e') = (a+b/2+c/4,0,b/2+c/2+d/2,0,c/4+d/2+e)$. 
For pattern $B$, $(a',b',c',d',e') = (a,b+c/2,c/4+d/2,0,c/4+d/2+e)$. 
For pattern $C$, $(a',b',c',d',e') = (a+b/2+c/4,0,b/2+c/4,c/2+d,e)$. 
[Right] Example of a five column Pachinko with
three nonzero outputs.}
\label{fig:few}
\end{figure}
 
These patterns correspond respectively to the following three linear transformations: 
 
 \begin{equation}
 \small
A = 
 \left(
 \begin{array}{ccccc}
 1 & 1/2 & 1/4 & 0 & 0\\
 0 & 0 & 0 & 0 & 0 \\
 0 & 1/2 & 1/2 & 1/2 & 0\\
  0 & 0 & 0 & 0 & 0 \\
  0 & 0 & 1/4 & 1/2 & 1
\end{array}
 \right), B = 
  \left(
 \begin{array}{ccccc}
 1 & 0 & 0 & 0 & 0\\
 0 & 1 & 1/2 & 0 & 0 \\
 0 & 0 & 1/4 & 1/2 & 0\\
  0 & 0 & 0 & 0 & 0 \\
  0 & 0 & 1/4 & 1/2 & 1
\end{array}
 \right), C = 
  \left(
 \begin{array}{ccccc}
 1 & 1/2 & 1/4 & 0 & 0\\
 0 & 0 & 0 & 0 & 0 \\
 0 & 1/2 & 1/4 & 0 & 0\\
  0 & 0 & 1/2 & 1 & 0 \\
 0 & 0 & 0 & 0 & 1
\end{array}
 \right).
 \end{equation}
 
Assuming that the nonzero probability outputs in columns one and five, the sequence of pins may not end with a pin in column three, or with configuration $B$ or $C$, or else the Pachinko would have four outputs. So the pins must end with an $A$ configuration. Thus, the producible output distributions  are exactly the distributions that can be written in the following form, for a sequence of configurations $T_i \in \{A, B, C\}$ and a ball input vector $x\in \mathcal{X}$:
 
 \begin{equation}
 A\left(\prod\limits_{i = 1}^k  T_i\right)x.
 \end{equation}
 An example of one such Pachinko is shown on the right of Figure~\ref{fig:few}.

\subsection{Uniform Distributions}

\begin{figure}[h!]
\centering
\includegraphics[width=5 in]{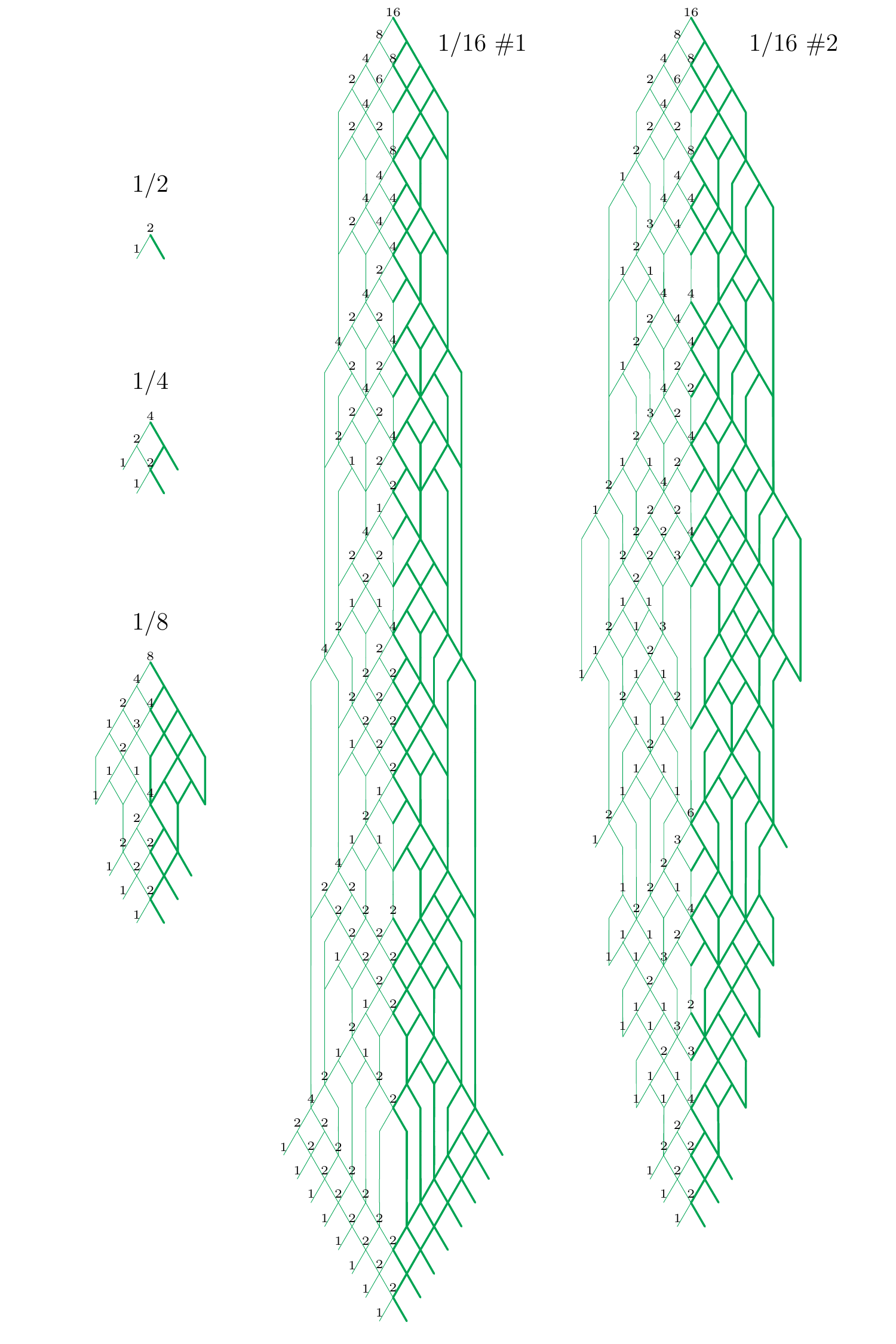}
\caption{50-50 Grid Pachinkos that construct uniform distributions for 2, 4, 8, and 16 outputs, two provided for the last.
Pins are at each vertex with out-degree (down-degree) two. The left half of each tracks how probability
mass travels through the distribution. The probabilities are the numbers shown divided by the number of outputs.}
\label{fig:uniform}
\end{figure}

Uniform distributions are another special class of outputs that one might like to construct under this model, particularly for liquid distribution applications. Figure~\ref{fig:uniform} shows 50-50 Pachinkos that output uniform distributions of $1/2$, $1/4$, $1/8$, and $1/16$ respectively. Two different 50-50 Pachinkos with a uniform distribution of $1/16$ are given, demonstrating that different distributions of pins can yield the same output. It is still open if all uniform distributions of probabilities of the form $1/2^k$ are constructible by 50-50 Pachinkos.

\begin{open}
  Is every uniform distribution of output probabilities of the form $1/2^k$ constructible by a 50-50 Pachinko?
\end{open}

\subsection{Any Probability Distribution, Shifted}

While one cannot output every dyadic probability distribution by itself using a 50-50 Pachinko, one can output a multiple of any probability distribution, in addition to some other probabilities. 


%

\begin{theorem}
For any ordered set of dyadic probabilities $p_1$, $p_2$, \ldots, $p_n$ there is a positive constant $c$ and a 50-50 Pachinko  for which the probabilities for $n$ of its outputs, in order from left to right, are $cp_1$, $cp_2$, \ldots, $cp_n$. 
\end{theorem}

\begin{proof}
Let $m$ be the maximum of the numerators of the $p_i$s. We will construct an example with $c = \frac{3}{2^{2^{m+2n+10}}}$. 

For each $i\in\{1,\ldots,n\}$, let $cp_i = 3q'_i/r'_i$ in lowest terms.  For each $i\in\{1,\ldots,n\}$, $q'_i \le m < 2^{m+2i+4}$, so by doubling numerators and denominators as necessary, let $cp_i = 3q_i/r_i$ where $r_i$ is a power of 2 and $2^{m+2i+4} < q_i \le 2^{m+2i+5}$. That is, we put the numerators in increasing order and spaced out from each other and from 1.

We create two long diagonals of Pachinko pins as shown on the left of Figure~\ref{fig:shift}.

\begin{figure}[hpbt]
\centering
\includegraphics[width=6 in]{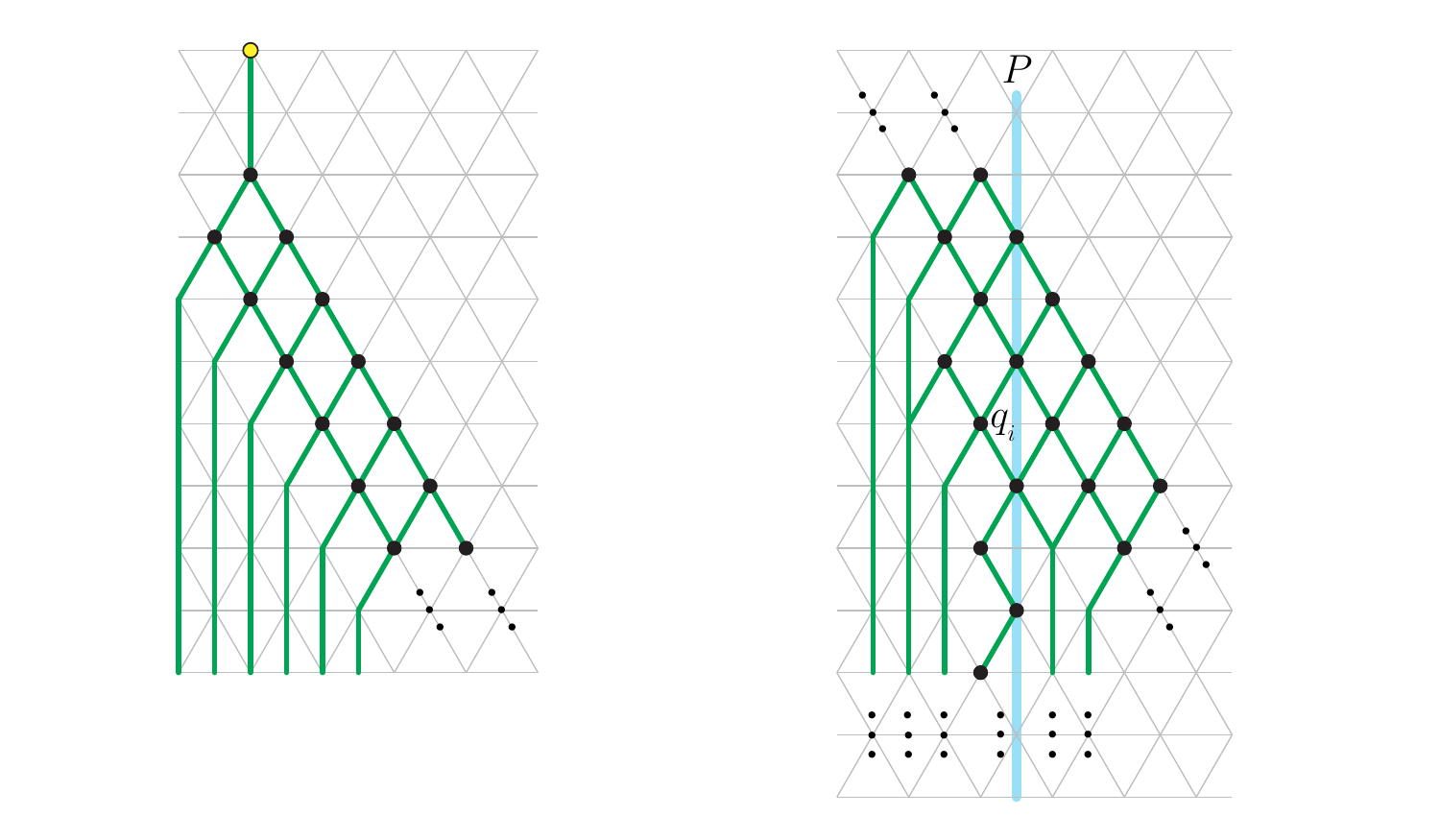}
\caption{Details of a Pachinko to construct a shifted distribution.
[Left] Top are consisting of two diagonal rows of pins. 
[Right] The vicinity of an output column $P$. }
\label{fig:shift}
\end{figure}

The probability that a ball hits the $k$th pin in the first diagonal is $2^{1-k}$, and the probability that a ball hits the $k$th pin in the second diagonal is $k2^{-k}$. 

For each $i$, put three pins in the third diagonal in spots $q_i - 1$, $q_i$, and $q_i + 1$, so the probabilities that balls hit them are $(q_i - 1)2^{-q_i}$, $(2q_i - 1)2^{-1-q_i}$, and $3q_i2^{-2-q_i}$, respectively. Since $q_{i-1} < q_i - 3$ and $q_{i+1} > q_i + 3$, they don't interfere.

Finally, the left output of the last of those three pins, $P$ (which is hit with probability $3q_i2^{-2-q_i}$), goes to a spot where no other output goes, so by alternating between pins in that column and $P$'s column, we can get a probability of the form $3q_i 2^{-2-q_i-t}$ for any $t$. In particular, since $2 + q_i \le 2^{m+2i+5} + 2 < 2^{m+2n+10}$, we can get $\frac{3q_i 2^{-2^{m+2n+10}}}{r_i} = cp_i$, as desired.
\end{proof}

\subsection{Approximating Any Probability Distribution}

While outputting any single dyadic probability with 50-50 Pachinkos remains open, one can approximate any real probability distribution to arbitrary precision. 

\begin{theorem}
For any real probability $p$ and any $\varepsilon > 0$, there exists a 50-50 Pachinko with a finite number of pins that outputs a probability in $(p - \varepsilon, p + \varepsilon)$.
\end{theorem}

\begin{proof}
Let $m$ and $n$ be integers such that $|\frac{m}{n} - p| \le \frac{\epsilon}3$. Without loss of generality assume $n > m > 0$.

By Lemma~\ref{lem:one}, there exists a finite arrangement of pins that achieves probability at least $1 - \frac{\varepsilon}{3\max(|m|,|m-n|)}$ in column $m$, leaving at most $\frac{\varepsilon}{3\max(|m|,|m-n|)}$ in all other columns combined, and in particular at most $\frac{\varepsilon}{3\max(|m|,|m-n|)}$ between columns $m-n$ and $m$, not inclusive. Remove all pins in column $m-n$ from that arrangement. Doing this does not increase the probability that a ball ends between columns $m-n$ and $m$, exclusive, since it affects only balls that make it to column $m-n$, and decreases the probability that such a ball ends between columns $m-n$ and $m$ to zero. So the ball ends between $m-n$ and $m$, exclusive, with probability at most $\frac{\varepsilon}{3\max(|m|,|m-n|)}$, and otherwise ends in column $m-n$ or $m$. But the first moment is 0, so if $p_i$ is the probability that the ball ends in column $i$, then $\sum i p_i = 0$, so $|mp_m + (m-n)p_{m-n}| = |\sum_{i \in (m-n,m)} ip_i| < \max(|m|,|m-n|)\frac{\varepsilon}{3\max(|m|,|m-n|)} = \frac{\varepsilon}{3}$. Also, the total probability is 1, so $|p_m + p_{m-n} - 1| \le \frac{\varepsilon}{3\max(|m|,|m-n|)} \le \frac{\varepsilon}{3m}$. By the triangle inequality,  $|mp_m + (m-n)p_{m-n}| = |m(p_m + p_{m-n}) - np_{m-n}| \ge |m - np_{m-n}| - \frac{\varepsilon}{3}$, so $|m - np_{m-n}| \le \frac{2 \varepsilon}{3}$, so $|\frac{m}{n} - p_{m-n}| \le \frac{2\varepsilon}{3}$, so by the triangle inequality again, $|p - p_{m-n}| \le \varepsilon$, as desired.
\end{proof}

\begin{theorem}
For any finite ordered set of real probabilities $(p_1,\ldots,p_n)$ and
any $\varepsilon > 0$, there exists a 50-50 Pachinko with a finite
number of pins for which the probabilities for $n$ of its consecutive
outputs, in order from left to right, are in $(p_i - \varepsilon, p_i + \varepsilon)$ for $i\in \{1,\dots,n\}$ .
\end{theorem}
\begin{proof}
Note that the previous
proof proves the stronger statement that for any real probability $p$,
any $\varepsilon > 0$, and sufficiently large $n$, there exists a
50-50 Pachinko with a finite number of pins that takes probability 1 at
column $i$ to a probability distribution within statistical distance\footnote{
By \emph{statistical distance}, we mean the maximum absolute difference between
respective probabilities of the corresponding distribution.
}
$\varepsilon$ of probability $p$ at column $i-n$ and probability $1-p$
at some column $j \ge i+n\frac{p}{1-p}$, without using any pins at columns outside the
interval $(i-n,j)$. 

We will prove the claim by induction on $n$, essentially by
applying the above in series many times. 

The base case $n=1$ is the previous theorem. If $n\geq2$, consider the
inductive construction of the probability distribution
$(\frac{p_1}{1-p_n},\ldots,\frac{p_{n-1}}{1-p_n})$ to within
$\frac{\varepsilon}{3}$. It uses finitely many pins, so they all lie
between columns $-n$ and $n$ for some $n$. So, first construct a
distribution within statistical distance $\frac{\varepsilon}{3}$ of
probabilities $1-p_n$ and $p_n$ separated by at least $2n$ columns. Then
apply that inductive construction to the output with probability
$1-p_n$, giving probabilities $p_1$, \ldots, $p_{n-1}$ in consecutive
columns, without interfering with the probability of $p_n$ more than $n$
columns away. Finally, apply Lemma~\ref{lem:one} to take the probability
$p_n$ to the column just right of $p_{n-1}$, again to within statistical
distance of $\frac{\varepsilon}{3}$, giving $(p_1,\ldots,p_n)$ to within
statistical distance $\varepsilon$, as desired. (If $n=2$, the last step
is the only relevant step.)
\end{proof}

%
%



%
%
%
%
%

%
%
%

%
%
%
%


\section{General Pachinkos}
\label{sec:general}

\subsection{Universality}

For general Pachinkos, any ordered set of dyadic probabilities may be constructed.

\begin{theorem}
Every finite ordered set of dyadic probabilities are exactly the ordered outputs of some Pachinko.
\end{theorem}

\begin{proof}
We can trivially construct a Grid Pachinko that outputs any finite ordered set of probabilities. 
Construct a complete balanced binary tree with each leaf having probability equal to the 
lowest bit in the input set. This tree can be constructed by using N-Pins at junctions and
L-Pins and R-Pins to move the junctions apart as necessary. Then use L-Pins and R-Pins
to combine the leaf output bits together to form the required outputs. Since Grid Pachinkos
can be simulated by General Pachinkos, the claim holds. Of course this construction requires 
the use of an exponential number of pins.
\end{proof}

Thus, the question arises, how many pins does a Pachinko need to output a given dyadic probability distribution? Let $(p_1,p_2,...p_n)$ be the target dyadic probability distribution with each $p_i$ the probability of the ball terminating in column $i$.
Let $2^k$ be the largest denominator in any $p_i$, so each probability may be represented using $k$ bits. We show that the constructing an arbitrary dyadic distribution requires at least $\Omega(nk)$ pins, and we give an algorithm that constructs the target distribution using $O(nk^2)$ pins.

\subsection{Lower Bound on Pins}

\begin{theorem}
For any $n$ and $k$ there exists an (unordered) set of $n$
$k$-bit dyadic output probabilities requiring $\Omega(nk - n \log n)$ Pachinko
pins.
\end{theorem}

Note that since each probability is at least $\frac{1}{2^k}$, there are
at most $2^k$ nonzero probabilities, so for reasonable parameter choices
the first term dominates and gives a lower bound of $\Omega(nk)$ pins.

\begin{proof}
The number of ordered such probability distributions is
$\Omega((\frac{2^k}{n})^n)$, since there are that many probability
distributions in which the first $n-1$ probabilities are chosen
independently between $0$ and $\frac1n$.

So the number of ordered probability distributions is $\Omega(2^{nk - n
  \log n})$, and the number of unordered probability distributions is
also $\Omega(2^{nk-n\log n})$, since each unordered probability
distribution corresponds to at most $n! = O(2^{n \log n})$ ordered
ones. 

A Pachinko arrangement is completely specified by the planar digraph
describing which pins' outputs are which pins' inputs. Bonichon et al.~\cite{BGHPS}
showed that the number of (unlabeled) planar graphs on $t$ vertices is $O(2^t)$. A
planar digraph consists of a planar graph and an orientation for each of
its $O(t)$ edges, so the number of planar digraphs on $t$ vertices is
also $O(2^t)$. So at least $nk - n \log n$ pins are necessary to even
construct enough planar digraphs to generate all the probability
distributions.
\end{proof}

\subsection{Upper Bound on Pins}
In this section, we prove an upper bound of $O(nk^2)$ on the same
problem as a lower bound was given for in the previous section.

We define a \textit{pure stream} to be the Pachinko region for which the
probability that the center of the ball passes through that region is $1/2^m$ for any integer $m$.
An \textit{impure stream} is a region that has
$\sum_{i=1}^{m}(1/2^{a_i})$ of  probability that the center of the ball
passes through that region with $(a_1,\ldots,a_m)$ being a strictly increasing integer sequence.

\begin{lemma}
\label{lem:crossover}
A pure stream $1/2^m$ and a impure stream $\sum_{i=1}^{j}(1/2^{a_i})$,
with $m<a_1$ and separated by 3 empty columns, can switch columns using $O(a_j-m)$ pins.
\end{lemma}
\begin{proof}
We provide a constructive example in Figure \ref{fig:newCross}.
Without loss of generality, consider that the pure stream is positioned on the left.
First, move the pure stream to the center (two columns to the right) using two R-pins.
Then, use an N-pin to split the pure stream into two streams of value $1/2^{m+1}$. 
Now, we check if $1/2^{m+1}\in (a_1,\ldots,a_j)$. 
If that is the case, we add two R-pins: one to direct one of the pure streams to the center and the other to merge it with the impure stream to the right.
Otherwise, use two L-pins at the same position.
Repeating the splitting process until the pure stream at the center is split into two streams of value $1/2^{a_j}$ ($a_j-m-1$ times). 
Now, use an L-pin and an R-pin as shown in Figure \ref{fig:newCross}.

We call this construction a \textit{crossover gadget}.
With this gadget, we can construct an impure stream of value $\sum_{i=1}^{j}(1/2^{a_i})$ in the left column, and a pure stream of value $1/2^m$ in the right column.
The construction uses $2+3(a_j-m)$ pins.
\end{proof}

\begin{figure}
\centering
	\includegraphics[width=6in]{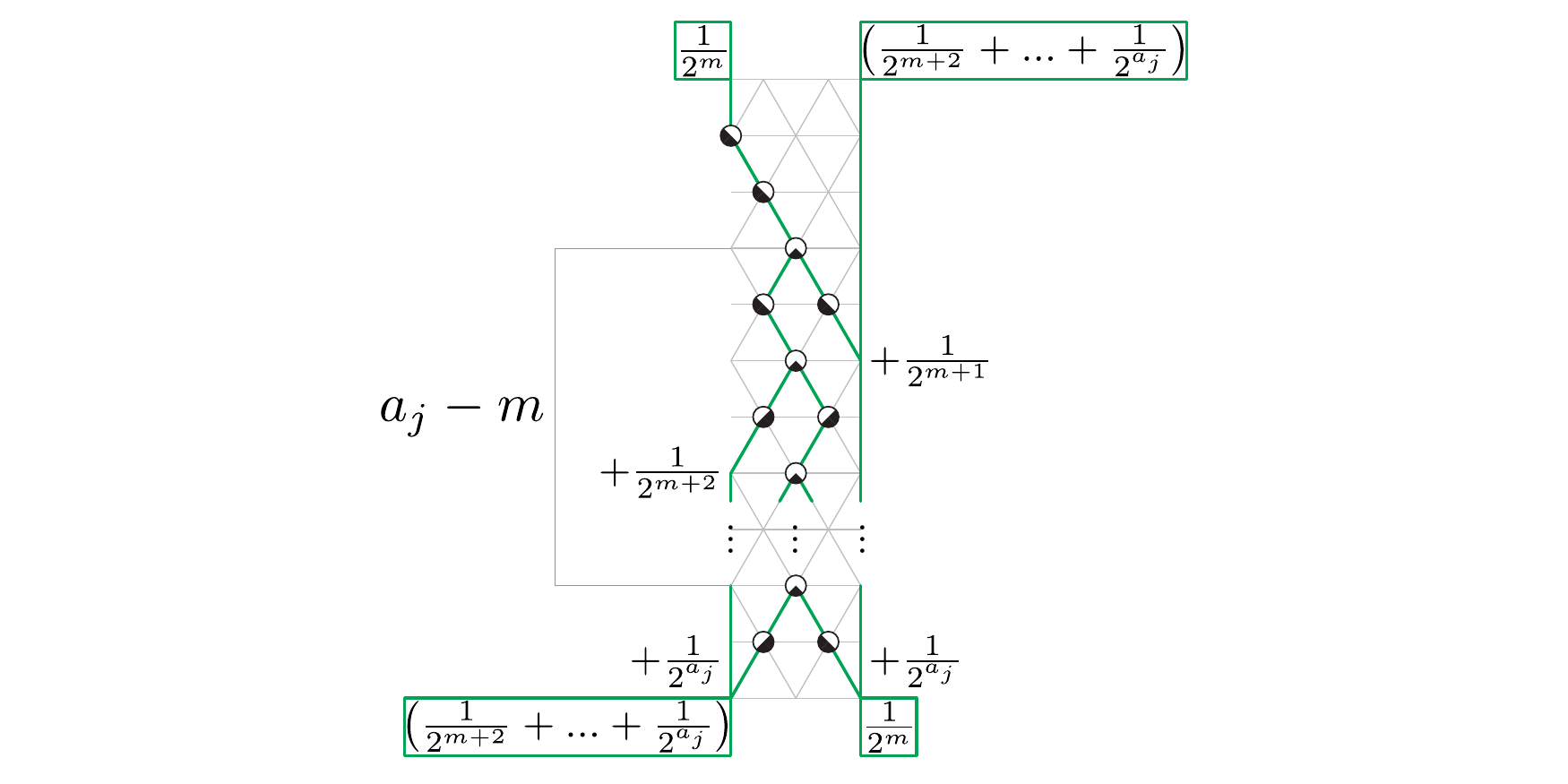}
    \caption{Crossover gadget.}
    \label{fig:newCross}
\end{figure}

\begin{lemma}
\label{lem:dyadicProbability}
Any dyadic probability, $\sum_{i=1}^{j}(1/2^{a_i})$ with $(a_1,\ldots,a_j)$ strictly increasing, can be constructed with $O(k^2)$ pins from a set of $O(k)$ pure streams, with ordered probabilities from higher to lower and with $O(1)$ distance between neighboring streams, if, for every $a_i$, there is a pure stream with probability $1/2^{a_i}$.
\end{lemma}
\begin{proof}
Begin by adjusting the distances between the pure streams, using R-pins and L-pins so that neighboring streams have 3 empty columns between them.
This takes $O(k^2)$ pins, since each pure stream can move $O(k)$ units.
We define the stream of probability $1/2^{a_j}$ as the \textit{working stream}. 
We will bring it to the left using L-pins.
If the pure stream immediately to the left of it has probability $1/2^{a_i}$ for some $i$ and the next pure stream to the left has a greater probability, merge the $1/2^{a_i}$ stream with the working stream. 
Otherwise, use a crossover gadget to bring the working stream to the left.
When there are no more pure streams to the left, the working stream will have probability $\sum_{i=1}^{j}(1/2^{a_i})$.
Each crossover costs $O(k)$, by Lemma \ref{lem:crossover}, and the number of crossovers needed is $O(k)$.
Therefore, the number of required pins is $O(k^2)$.
\end{proof}

\begin{theorem}
\label{dyadicDistribution}
Any dyadic probability distribution $(p_1, p_2, \ldots, p_n)$ can be
obtained with $O(nk^2)$ pins, in which $k$ is the maximum number of bits
used to represent a probability in the desired distribution.
\end{theorem}
\begin{proof}
This is a proof by construction.
Begin by splitting the initial stream (of probability 1) into $k-1$ streams of probabilities $(1/2^i)_{i=1}^{k-1}$ ordered by decreasing order from left to right followed by two streams of probability $ 1/2^k $.
This can be done by successively splitting the rightmost stream with an N-pin, $k$ times.

Using Lemma \ref{lem:dyadicProbability}, we can obtain $p_1$ using $O(k^2)$ pins.
Notice that this construction decreases the number of pure streams and that Lemma \ref{lem:dyadicProbability} cannot be applied directly to obtain $p_2$, because a term $1/2^m$ of $p_2$ might not have a corresponding pure stream.
We then re-split the pure streams even further using the following algorithm.

Let $1/2^p$ be the smallest term of $p_2$ that does not have a corresponding pure stream, and let $1/2^q$ be the smallest pure stream such that $p>q$.
Divide the stream $1/2^q$ with N-pins in the same way as the initial stream was split, $p-q$ times, so that the resultant rightmost stream will have probability $1/2^p$.
This procedure maintains the order of probabilities.
Repeat until all terms of $p_2$ have corresponding streams.
This subroutine uses $O(k)$ pins for each iteration. Overall, $O(k^2)$ pins are used, because $p_2$ has $O(k)$ bits.
After this re-splitting phase, the number of pure streams cannot be greater than $2k$.
If it were, there would be more than 2 streams with the same probability, because the smallest probability that can be created is $1/2^k$. One of them would have been created in the re-splitting phase, which causes a contradiction since we only create probabilities that were previously inexistent.

The subroutine described above allows the construction of $p_2$ with $O(k^2)$ pins, by applying Lemma \ref{lem:dyadicProbability}.
All statements referring to $p_2$ can be generalized to $p_i$, $i>2$.
We can then construct all probabilities with a total of $O(nk^2)$ pins.
\end{proof}


\section{Conclusion}
In this paper we have analyzed a rich set of models of perfectly inelastic Pachinko
with a single ball as input. A natural open question is to extend these results for Pachinko
allowing multiple input balls. For Grid Pachinkos, we allow for pins that stop, move or evenly split the forward movement of a ball. Another natural generalization arises by allowing for $(\alpha,\beta)$-pins, pins where the ball will either move to the right or left with probability $\alpha\beta$ or $(1-\alpha)\beta$ respectively and continue moving down or remain at the pin with probability $1-\beta$. All our Grid Pachinko pins are special cases of this model. Modeling elastic, non-local behavior, where balls may move up and/or more to the left or right, is perhaps the more natural, yet more complex system to analyze. Additionally, our models consider Pachinko in a vertical plane. 
Modeling three-dimensional balls interacting with point or line obstacles in the
presence of gravity may also lead to interesting future work. 
\section{Acknowledgements}
The authors would like to thank Zachary Abel, Greg Aloupis, Nadia Benbernou, Scott Kominers, and Anika Rounds for helpful discussions. 

\bibliography{pachinko}

\begin{thebibliography}{1}

\bibitem{abramowitz1966handbook}
Milton Abramowitz, Irene~A Stegun, et~al.
\newblock Handbook of mathematical functions.
\newblock {\em Applied Mathematics Series}, 55:62, 1966.

\bibitem{Akiyama-Ruiz-2008}
Jin Akiyama and Mari-Jo~P. Ruiz.
\newblock Pachinko math.
\newblock In {\em A Day's Adventure in Math Wonderland}. World Scientific,
  2008.

\bibitem{BGHPS}
Nicolas Bonichon, Cyril Gavoille, Nicolas Hanusse, Dominique Poulalhon, and
  Gilles Schaeffer.
\newblock Planar graphs, via well-orderly maps and trees.
\newblock {\em Graphs and Combinatorics}, 22:185--202, 2006.

\bibitem{de2000computational}
Mark De~Berg, Marc Van~Kreveld, Mark Overmars, and Otfried~Cheong Schwarzkopf.
\newblock {\em Computational geometry}.
\newblock Springer, 2000.

\bibitem{Eames}
Eames{ Office}.
\newblock Mathematica: A world of numbers{\dots} and beyond.
\newblock \url{http://www.eamesoffice.com/the-work/mathematica/}.

\bibitem{Graham}
R.~L. Graham, M.~Gr\"{o}tschel, and L.~Lov\'{a}sz, editors.
\newblock {\em Handbook of Combinatorics (Vol. 2)}.
\newblock MIT Press, Cambridge, MA, USA, 1995.

\bibitem{JapanZone}
JapanZone.
\newblock Pachinko.
\newblock \url{http://www.japan-zone.com/modern/pachinko.shtml}.

\bibitem{Wiki}
Wikipedia.
\newblock Pachinko.
\newblock \url{http://en.wikipedia.org/wiki/Pinball}.

\end{thebibliography}
\bibliographystyle{plain}

\end{document}